\documentclass[12pt]{article}
\usepackage[margin=0.9in]{geometry}
\usepackage{booktabs} % For formal tables
\usepackage[ruled]{algorithm2e} % For algorithms
\usepackage{amsmath}
\usepackage{amsthm}
\usepackage{caption}
\usepackage{subcaption}
\usepackage{amssymb}

\usepackage{thmtools}
\usepackage{thm-restate}

\usepackage{comment}
\usepackage{xspace}
\usepackage[usenames,dvipsnames]{xcolor}
\usepackage[colorlinks=true,citecolor=ForestGreen,linkcolor=ForestGreen,urlcolor=NavyBlue]{hyperref}
\usepackage{cleveref}

\usepackage{tikz}
\usetikzlibrary{arrows,positioning}

\newtheorem{theorem}{Theorem}[section]
\newtheorem{lemma}[theorem]{Lemma}

\newtheorem{definition}[theorem]{Definition}
\newtheorem{corollary}[theorem]{Corollary}
\newtheorem{innercustomthm}{Theorem}
\newenvironment{customthm}[1]
  {\renewcommand\theinnercustomthm{#1}\innercustomthm}
  {\endinnercustomthm}

\DeclareMathOperator*{\argmax}{arg\,max}
\DeclareMathOperator*{\argmin}{arg\,min}

\SetAlFnt{\small}
\SetAlCapFnt{\small}
\SetAlCapNameFnt{\small}
\SetAlCapHSkip{0pt}
\IncMargin{-\parindent}

\newcommand{\p}{\mathbf{p}}
\usepackage{natbib}

\usepackage{footnote}
\makesavenoteenv{algorithm}

\newcommand{\alloc}{\mathbf{a}}
\newcommand{\x}{\mathbf{x}}
\newcommand{\eq}{\mathsf{EQ}}
\newcommand{\E}{\mathbb{E}}
\newcommand{\on}{\text{on}}
\newcommand{\off}{\text{off}}

\setcitestyle{authoryear}

\renewenvironment{abstract}
 {\small
  \begin{center}
  \bfseries \abstractname\vspace{-.5em}\vspace{0pt}
  \end{center}
  \list{}{
    \listparindent 1.5em%
    \setlength{\leftmargin}{8mm}%
    \itemindent    \listparindent
    \setlength{\rightmargin}{\leftmargin}%
  }%
  \item\relax}
 {\endlist}

\author{
  Alon Eden\\
  Hebrew University of Jerusalem\\
  \texttt{alon.eden@mail.huji.ac.il}
  \and
  Gary Qiurui Ma\\
  Harvard Univeristy\\
  \texttt{qiurui\char`_ma@g.harvard.edu}
  \and
  David C. Parkes \\
  Harvard University \\
  \texttt{parkes@eecs.harvard.edu}
}

\title{Platform Equilibrium: Analyzing Social Welfare in Online Market Places\thanks{The work of A. Eden was supported by the Israel Science Foundation (grant No. 533/23).}}
\date{}

\begin{document}
\maketitle

\begin{abstract}
We introduce the theoretical study of a \emph{Platform Equilibrium} in a market with unit-demand buyers and unit-supply sellers. Each seller can join a platform and transact with any buyer or remain off-platform and transact with a subset of buyers whom she knows. Given the constraints on trade, prices form a competitive equilibrium and clears the market. The  
platform charges a \emph{transaction fee} to all on-platform sellers, in the form of a fraction of on-platform sellers' price. The platform chooses the fraction to maximize revenue. A {\em Platform Equilibrium} is a Nash equilibrium of the game where each seller decides whether or not to join the platform, balancing the effect of a larger pool of buyers to trade with, against the imposition of a transaction fee.

Our main insights are:(i) In homogeneous-goods markets, pure equilibria always exist and can be found by a polynomial-time algorithm; (ii) When the platform is unregulated, the resulting Platform Equilibrium guarantees a tight $\Theta(\log (\min\{m,n\}))$-approximation of the optimal welfare in homogeneous-goods markets, where $n$ and $m$ are the number of buyers and sellers respectively; (iii) Even light regulation helps: when the platform's fee is capped at $\alpha\in[0,1)$, the price of anarchy is $(2-\alpha)/(1-\alpha)$ for general markets. For example, if the platform takes 30\% of the seller's revenue, a rather high fee, our analysis implies the welfare in a Platform Equilibrium is still a $0.412$-fraction of the optimal welfare. Our main results extend to markets with multiple platforms, beyond unit-demand buyers, as well as to sellers with production costs.

\end{abstract}

\section{Introduction}
\label{sec:intro}

During the COVID-19 crisis, entire populations were put under stay-at-home mandates and many businesses closed due to government restrictions. Nonetheless, some businesses weathered the situation well, even reporting a positive impact~\citep{bloom2021impact}. One of the factors that has been credited with alleviating the issue is internet economy. Specifically, online platforms such as Amazon and DoorDash were able to bridge the gap created by consumers' inability to directly transact with firms~\citep{raj2020covid}. At the same time, platforms used their increased market power to set high fees for merchants, leaving some  with very low or even negative profit margins~\citep{mckinsey2021}. As a response, states like New York and California, imposed caps on commissions, prompting some platforms to shift fees towards customers.\footnote{\url{https://www.protocol.com/delivery-commission-caps-uber-eats-grubhub}} We see that platforms have an increasing role in facilitating trade, but are also seen to act strategically in setting fees in order to maximize their gain. In this paper,  we introduce a new  model that  gives a  clean theoretical framework to  explain the interplay between the two, and provide insights into how platform regulation can play a role.

We consider a market with a set of \emph {unit-demand} buyers and \emph{unit-supply} sellers. Each buyer has a set of sellers they are able to directly transact with off-platform. Each seller can join a trading platform to transact with any buyer. Given the transaction constraints dependent on sellers' decisions, prices and transactions are induced by a competitive (Walrasian) equilibrium. For a seller, joining the platform weakly increases her price as a result of improved access to buyers, but she must also pay a {\em transaction fee} to the platform, which takes the form of a fraction of her revenue. The fraction is between 0 and 1, and chosen by the platform. Given a transaction fee fraction, a set of sellers each buyer can trade with off-platform, and buyers' valuations, a platform game is defined where each seller {simultaneously}
chooses whether or not to join the platform (or the probability with which they join). We name the Nash equilibrium that is formed in this game a {\em Platform Equilibrium}. We define the \emph{social welfare} as the sum of transacting buyers' valuations, and \emph{optimal welfare} as the social welfare when all sellers join the platform. 

We first study the existence of pure Nash equilibria in this game, and how to compute an equilibrium efficiently. We then study ratio of social welfare to optimal welfare, when the platform sets the transaction fee fraction to maximize its own revenue in the induced Platform Equilibrium. We continue to analyze the effect of regulation on social welfare, modelled by a cap on the transaction fee fraction. Finally, we extend our analysis to markets with multiple platforms, and beyond unit-demand buyers, and also consider sellers with production costs.

\subsection{Results and Techniques}
\label{sec:results}

\paragraph{Existence of pure Nash equilibrium and computation.}
In Section~\ref{sec:existence_of_pure}, we explore the existence of pure Platform Equilibria, in the platform game under a fixed transaction fees and market conditions. Throughout the paper, we assume the market clears at the maximum competitive prices. {This reflects a model in which sellers have  market power, so that it is sellers and not buyers that are, in effect, setting prices for trades (subject to the standard requirements of competitive equilibrium, which require balance of supply and demand).} 
As shown in Appendix~\ref{app:prices-min}, some of our results cease to hold with other choices for
competitive equilibrium prices, such as the minimum competitive equilibrium prices. We show when a buyer values all sellers' items equally, also known as the {\em homogeneous-goods} case, pure Platform Equilibrium always exist, and can be found by a simple procedure. The existence of a pure equilibrium is no longer guaranteed in general markets.

\begin{customthm}{1}[Informal Version of Proposition~\ref{prop:no_pure} and Theorems~\ref{thm:PE_pure}]\label{thm:informal_PE_pure}
    In any homogeneous-goods market, a pure Platform Equilibrium exists for any transaction fee fraction and can be found in polynomial time by Algorithm~\ref{alg:PE_pure}. In contrast, there exists general markets where pure Equilibrium does not exist for some values of transaction fee fraction.
\end{customthm}

The proof of this result is involved and draws upon the unique structure for homogeneous-goods market. \citet{kranton2000competition} defines the notion \emph{opportunity path}, which alternates between buyers and sellers to capture competitive prices for sellers' items. We use this combinatorial structure to decompose social welfare into optimal welfare achieved by sellers off-platform, and the remaining on-platform component. This in turn offers a simple expression for sellers' on-platform prices. Finally, we reason combinatorially to show as the transaction fee fraction lowers from 1 to 0, Algorithm~\ref{alg:PE_pure} can find pure equilibrium with any number of on-platform sellers.

The above theorem implies different ways to study social welfare in homogeneous-goods and general markets. In the former, we can use pure equilibria to bound platform revenue and social welfare. In the latter we need to reason about mixed equilibria when studying the effect of a platform on the market. 

\vspace{0.1in}

After studying the existence and computational aspects of Platform Equilibria, in Sections~\ref{sec:poa_rev_max} and \ref{sec:poa-regulated}, we analyze how the platform affects the social welfare in the market. We consider a platform setting the transaction fee fraction strategically to maximize its own revenue, and the sellers respond by playing a Platform Equilibrium. Since there can be multiple Platform Equilibria, we assume the platform can use its market power 
to choose a favorable equilibrium, given its choice of transaction fee fraction. We stress our results for Section~\ref{sec:poa_rev_max} and \ref{sec:poa-regulated} are proved for mixed as well as pure equilibrium.
\vspace{0.1in}

\noindent\textbf{Social welfare with an unregulated platform.}
As a first study, in Section~\ref{sec:poa_rev_max}, we consider the case where the platform is unregulated in setting fees. Let $n$ and $m$ be number of buyers and sellers, respectively. In a homogeneous-goods market, the obtained welfare can be as low as $O(1/\log (\min\{n,m\}))$-fraction of the optimal welfare, and this is tight. We show tightness by showing the platform is able to extract $\Omega(1/\log (\min\{n,m\}))$-fraction of optimal welfare as revenue. This is done by analyzing different possible market structures, and showing how the  platform extracts the desired revenue target in each one. As platform revenue provides a lower bound on the welfare, this yields the desired bound. In the proof, we utilize the pure equilibria finding algorithm to show how to choose a fee and a corresponding equilibrium to extract enough revenue (recall that we assume the platform can select a desirable equilibrium). The platform might choose a different fee and possibly mixed equilibrium to extract even more revenue, which implies the revenue guarantee of the platform, and as a result, the welfare guarantee of the buyers and sellers.

\begin{customthm}{2}[Informal Version of Theorem~\ref{thm:poa_upper_bound_homo} and Theorems~\ref{thm:poa_lower_bound_homo}]\label{thm:informal_poa_upper_bound_homo}
    When the  platform is unregulated, there exists a homogeneous-goods market in which the Platform Equilibrium is $O(1/\log(\min\{n,m\}))$-fraction of the optimal welfare. For every homogeneous-goods market, {the optimal fee for} a revenue-maximizing platform results in a Platform Equilibrium that guarantees at least $\Omega(1/\log(\min\{n,m\}))$-fraction of the optimal welfare.
\end{customthm}

The above result shows even though the platform is maximizing its own revenue, we still get welfare guarantees as a byproduct. However, this guarantee is modest, and only for a very restrictive market setting. For general markets, the prospect is worse: there exists a market where the Platform Equilibrium is only $O(1/\min\{n,m\})$-fraction of the optimal welfare. This motivates us to study the effect of imposing regulations on the fees that a platform can charge, something that is typical in practice \citep{FeeCapNews}. As we show, even light regulation can provide much stronger welfare guarantees, and even for general markets.

\paragraph{Social welfare with fee cap.} In Section~\ref{sec:poa-regulated}, we present the main positive result of the paper. Let $\alpha$ be the transaction fee fraction the platform posts and Price of Anarchy (PoA) of Platform Equilibrium be the ratio of optimal welfare to social welfare across all markets. We study PoA when regulators cap the highest possible $\alpha$ in general unit-demand unit-supply markets. Since social welfare is larger for lower transaction fees, we bound the PoA assuming the platform posts fee $\alpha$ exactly at this cap. If the platform posts a lower fee for higher revenue, this bound still holds. This analysis does not assume the platform uses market power to select a particular Platform Equilibrium, and thus work for all equilibria. Concretely, we show the following.

\begin{customthm}{3}[Informal Version of Theorem~\ref{thm:pure_poa}, \ref{thm:mixed_poa} and \ref{thm:poa_tight}]\label{thm:informal_poa}
    For any fixed fee $\alpha\in [0,1)$, the Price of Anarchy of Platform Equilibrium is at most $\frac{2-\alpha}{1-\alpha}$. Moreover, for every $\alpha$, there is a market for which the ratio of optimal welfare to social welfare is indeed $\frac{2-\alpha}{1-\alpha}$, and sellers use pure strategies.
\end{customthm}

The theorem implies if the platform's fee is capped at $30\%$, a quantity higher than most real world delivery platform fees as shown in Table~\ref{table:platforms and their coommission rate}, the resulting social welfare in every Platform Equilibrium is at least $41\%$ of the optimal welfare. Without the platform, the obtained welfare can be arbitrarily low.

As discussed above, a pure Platform Equilibrium might not exist for general markets. Thus, for this result we also analyze PoA when sellers employ mixed strategies.\footnote{As the strategy space is compact and the game is finite, a mixed Platform Equilibrium is guaranteed to exist.} Our proof takes a mixed Platform Equilibrium $\mathbf{x}=(x_1,\ldots,x_m)$, where $x_i$ is the probability seller $i$ joins the platform, and builds a
corresponding \emph{Bayesian game} for theoretical analysis. In this Bayesian game, each seller can always trade with its off-platform buyers in the original game, but with probability $x_i$ it can trade with all buyers. Each seller needs to decide whether to join the platform, and pay an $\alpha$ fee to the platform, or stay off the platform, \emph{before} it knows whether it can trade with all buyers. We show that if $\mathbf{x}$ is an equilibrium in the original game, then no sellers joining is an equilibrium of the new Bayesian game, and the expected welfare is at least $\frac{1-\alpha}{2-\alpha}$-fraction of the optimal welfare.

\begin{table}[t]
\centering
\begin{tabular}{|c|c c c c|} 
 \hline
 Platforms & Amazon & UberEats & DoorDash & Grubhub\\ [0.5ex] 
 \hline
 Commission Rate & 8\%-20\% & 15\%-30\% & 15\%-30\% & 15\%-25\%  \\  
 \hline
\end{tabular}
\caption[Platforms and their commission rate in the US from 2021-2022.]{Platforms and their commission rate in the US from 2021-2022.\protect\footnotemark}
\label{table:platforms and their coommission rate}
\end{table}
\footnotetext{This data is gathered from platforms' US websites in 2023. Amazon referral fees are charged according to item categories. 8\%-20\% for all categories except `Amazon device accessories'.
\url{https://sell.amazon.com/pricing}. For UberEats, DoorDash and Grubhub see \url{https://merchants.ubereats.com/us/en/pricing/},\url{https://get.doordash.com/en-us} 
and \url{https://get.grubhub.com/products/marketplace/}}

\paragraph{Generalizations.} Finally, we extend our model and results beyond the single platform unit-demand and zero cost setting. Section~\ref{sec:multi_platform} proves the Price of Anarchy results under regulation still hold in a market where multiple platforms compete for buyers and sellers. In Section~\ref{sec:seller_with_cost}, we show when sellers have production costs, the Price of Anarchy results extend naturally, with the results factoring out the cost. In Section~\ref{sec:beyond-ud}, we extend our results beyond unit-demand and consider additive-over-partition matroids valuation. That is, for each buyer, sellers are partitioned into a set of categories, and the buyer values at most a {\em capacity} number of sellers from each category. This generalizes both unit-demand and additive valuations, as well as $k$-demand valuations where a buyer is additive over their top $k$ items (e.g., in~\cite{BergerEF20,ZhangC20}).

\subsection{Related Work}
There are three works that use models similar to ours to examine how an online platform facilitates trade between buyers and sellers for revenue. \citet{birge2021optimal} study the optimal commission and subscription fee structure in a bipartite market. Also in a bipartite market, \citet{banerjee2017segmenting} consider a platform matching buyers to sellers for revenue. We share with these two works the idea that prices and trades are shaped by competitive forces endogeneous to the market, rather than being directly set by the platform. The main difference is that in these two works platform is the only venue where buyers and sellers interact, while we emphasize on the off-platform trading opportunities. This distinction allows us to model sellers' decision processes to join the platform. In fact, if off-platform options are not present in our model, sellers will always join the platform, resulting in arbitrarily high platform fees, while social welfare remains equal to optimal welfare. These two works are more general in other ways; for instance each node on the bipartite graph represents a continuum of buyers or sellers with varying valuations or costs. The third related work is \citep{d2024disrupting}, which models off-platform options. However, their focus is on platform matching buyers to sellers for revenue, whereas ours centers on platform setting fees.

Besides the three closest papers, our work build upon previous results on how network structure influence competitive equilibrium prices. \citet{kranton2000competition} study a homogeneous-goods buyer-seller market and relate the highest and lowest competitive prices to \emph{ opportunity paths} of trading agents. We make use of this structural result in the analysis in the present paper. \citet{elliott2015inefficiencies} generalizes the opportunity path argument to markets with general unit-demand valuations. \citet{kakade2004economic} further study the way in which equilibrium prices depend on the statistical structure of the buyer-seller bipartite graph. These works focus on price and network structure, and do not consider an intermediary as we do.   

Another strand of work that is related to ours is the study on network formation games with competitive equilibrium models of trade, subject to edge-constrained trading relationships. \citet{kranton2001theory} model buyers purchasing links to sellers in a bipartite economy before their valuations are realized, and show the Nash equilibrium of the network-formation game is fully efficient. Under a similar setting, \citet{even2007network} characterize all possible, equilibrium  bipartite graph topologies, based on the edge cost and the number of sellers and buyers, but do not discuss welfare. Similar to the two works, we also build on top of foundational studies on competitive equilibria in bipartite economies \citep{kelso1982job, gul1999walrasian}, and analyze the effect of market participants' incentives on network formation. But the platform-based emphasis leads our model to depart from theirs in various ways:
1) in our work, there is a revenue-optimizing platform that mediates trade and chooses transaction fees. In their work, sellers and buyers interact without a selfish mediator;  
2) as discussed in the previous paragraph, buyers and sellers in our work start with existing trading opportunities;
3) in our work, sellers acquire edges to all buyers simultaneously by joining the platform, and pay a fraction of the final competitive price for joining platform. In theirs, each edge has to be purchased individually with a fixed price per edge. 

Our work also relates to the broader literature that studies platform economics, which tends to ask questions about how to build network effects through subsidy under various forms of platform competition~\citep{caillaud2003chicken,armstrong2006competition,rochet2003platform}. Although our methods differ, we draw inspiration from this realm that a platform may subsidize one side of the market, while charging a fee to the other side (in our case, subsidizing buyers and charging sellers). For example, \citet{caillaud2003chicken}
study the  role  of platforms in matchmaking, with a platform that charges registration fees for entry and commission fees per transaction. Similarly, \citet{WangMELTZP23}, with registration and transaction fees, use simulation and reinforcement learning to investigate different regulatory approaches for platform prices to foster a more resilient economic system.

\section{Preliminaries}
\label{sec:prelims}

We adopt the buyer-seller network model of \citet{kranton2000competition}. There is a set of $n$ buyers,
 $B = \{1,\ldots, n\}$, and $m$ differentiated 
sellers, $S=\{1,\ldots, m\}$. Each seller has a single product to sell, and 
each buyer $i$ has a {\em unit-demand valuation}, $v_i$. Buyer $i$'s value for seller $j$'s product is $v_{ij}\geq 0$. A bipartite graph $G = \{g_{ij}\}_{\substack{i\in B\\ j\in S}}$ models which buyers and sellers can directly transact, without the use of the platform, where  
\begin{eqnarray}
    g_{ij} = \begin{cases} 1 \quad &\mbox{Buyer } i \mbox{ can directly transact with seller } j\\
                           0 \quad &\mbox{Otherwise}
            \end{cases}.
\end{eqnarray}

For buyer $i$,  $N(i) = \{j \ : \ g_{ij}=1\}$ denotes the list of sellers linked to the buyer in $G$. We define a competitive (Walrasian) equilibrium in the buyer-seller network model.
\begin{definition}[Competitive Equilibrium]
    A {\em competitive equilibrium} for the buyer-seller network $G$ is a tuple $(\p,\alloc)$. $\p=(p_1,\ldots, p_m)$ are none-negative item prices, $\alloc = \{a_{ij}\}_{\substack{i\in B\\ j\in S}}\in \{0,1\}^{n\times m}$ is an allocation of the goods to the buyers, and:
    \begin{itemize}
    \item Transactions must respect links: $a_{ij}\leq g_{ij}\quad \forall i\in B\ j\in S$.
    \item Buyers are allocated at most one good: $\sum_{j} a_{ij} \leq 1 \quad \forall i\in B$.
    \item Goods are sold at most once: $\sum_{i} a_{ij} \leq 1\quad \forall j\in S$.
    \item Buyers gets their most preferred outcome:  $u_i(\p,\alloc) \geq v_{ij} - p_j \quad \forall i\in B\ j\in S$. where 
    \begin{eqnarray}
    u_i(\p,\alloc) = \sum_{j}a_{ij}(v_{ij}-p_j),    \label{eq:buyer_util}
    \end{eqnarray}
    is buyer $i$'s utility for the allocation.
    \item Buyers have non-negative utility: $u_i(\p,\alloc)\geq 0\quad \forall i\in B$.
    \item Unassigned goods have price $0$.
    \end{itemize} 
\end{definition} \label{def:comp_eq}

A competitive equilibrium is a canonical model of the steady state in a market, capturing 
the notion of prices that are set such that supply meets demand. It follows from standard existence results~\citep{kelso1982job} that a competitive equilibrium always exists in a unit-demand buyer-seller network (a missing edge can be represented as $v_{ij}=0$). Moreover, competitive equilibria have a number of desirable properties. 
\begin{theorem}[First Welfare Theorem]
    In a competitive equilibrium, the social welfare is maximized with respect to 
the set of allocations that respect 
the transaction constraints posed by $G$.\label{thm:first_welfare} 
\end{theorem} 

\begin{theorem}[Second Welfare Theorem \citep{gul1999walrasian}]
    Let $(\p, \alloc)$ and $(\p',\alloc')$ be two competitive equilibria of a buyer-seller network $G$, then $(\p,\alloc')$ is also a competitive equilibrium (and so is $(\p',\alloc)$).
\end{theorem}

The Second Welfare Theorem implies that prices have the property of either forming a competitive equilibrium with any social-welfare optimal allocation, or forming a competitive equilibrium with none of them. We refer to prices $\p$ that are part of a competitive equilibrium as \emph{competitive prices}. It is also well known that 
 competitive prices   have a lattice structure.
\begin{theorem}[Lattice structure for competitive prices \citep{gul1999walrasian}]
    Let $\p_1$ and $\p_2$ be competitive prices, then $\p_1 \vee \p_2$ and $\p_1 \wedge \p_2$ are also  competitive prices, where $\vee$ is the coordinate-wise maximum and $\wedge$ is the coordinate-wise minimum.\label{thm: price_lattice}
\end{theorem}

As a result, there are {\em minimum} and {\em maximum} competitive prices, denoted $\underline{\p}$ and $\overline{\p}$ respectively,
with item-wise minimum and item-wide maximum prices
 denoted $\underline{p}_j$ and $\overline{p}_j$. 
For  $S$ sellers, $B$  buyers, network $G$,
 and buyer values $\mathbf{v}$, we use $W(S,B,\mathbf{v},G)$ to denote the 
{\em optimal welfare} from all feasible transactions between  sellers $S$ and  buyers $B$. When clear from context, we omit $\mathbf{v}$ from the notation. We will make use 
of the following characterization result.
\begin{theorem}[Characterization of competitive prices~\cite{gul1999walrasian}]
    The minimum and maximum competitive prices for an item $j$ have the following form:
\begin{align}
    \underline{p}_j &= W(S \cup \{j\}, B, G)-W(S, B, G), \label{eq:max_price}\\
    \overline{p}_j &= W(S, B, G) - W(S\setminus \{j\},B, G). \label{eq:min_price}
\end{align}

Here, $S \cup \{j\}$ denotes adding another copy of seller $j$ with all its edges to the market, and $S\setminus \{j\}$ is removing seller $j$ and all its connected edges from the market. The resulting graph changes correspondingly when adding or removing $j$ with all its edges, but for notation convenience we still use $G$ to denote the graph.
\end{theorem}

We next present our framework for analyzing the effect of the presence of a 
platform on a buyer-seller network.

\subsection{Platform Equilibrium} 
\label{sec:platform_eq}

Consider a buyer-seller network defined on graph $G$. Without a platform, a competitive equilibrium $(\p, \alloc)$ is formed subject to $G$ (see Definition~\ref{def:comp_eq}), where $\p$ are  competitive prices that support a welfare-optimal allocation, $\alloc$, subject to $G$. We do not assume anything regarding $\p$ except that it is competitive price vector. 

The utility of a buyer is given by Eq.~\eqref{eq:buyer_util}, while the utility of a seller $j$ is the revenue, $p_j$ (which is necessarily zero if the seller does not trade). By the First Welfare Theorem (Theorem~\ref{thm:first_welfare}), the obtained welfare in the market that forms in the absence of a platform is $W(S,B,G)$.  

Consider now a platform that declares a {\em transaction fee fraction}, $\alpha \in [0,1]$. Given that a set of sellers $P\subseteq S$ joins the platform, a new network $\hat{G}=G(P)= \{\hat{g}_{ij}\}_{\substack{i\in B\\ j\in S}}$ is formed with the following connections,
\begin{eqnarray}
    \hat{g}_{ij} = \begin{cases} 1 \quad & \mbox{Seller } j\in P\\
                           g_{ij} \quad & \mbox{Otherwise}
            \end{cases}.
\end{eqnarray}

That is, the constraints faced by a seller that does not join the platform are unchanged, 
whereas a seller that joins the platform can transact with any buyer.

Given the modified constraints on transactions, $\hat{G}$, a 
modified competitive equilibrium $(\hat{\p},\hat{\alloc})$ is formed, where $\hat{\alloc}$ is the optimal allocation constrained to $\hat{G}$, and $\hat{\p}$ are the supporting competitive prices. The utility of buyer $i$ in this competitive equilibrium 
 is $u_i(\hat{\p},\hat{\alloc})\geq 0$. The utility of a seller $j$ who does not join the platform is $\hat{p}_j$, while the utility of a seller $j$ who joins the platform is $(1-\alpha)\hat{p}_j$.\footnote{In reality, a seller joining the a  platform such as a food-delivery platform can still transact with a buyer off platform. 
But in the unit-supply unit-demand market, a seller has weakly no incentive to join the platform if it transacts directly with the buyer,  through an existing link. Therefore, we assume all sellers on platform pay transaction fee.} The platform's utility is $\alpha\cdot \sum_{j\in P}\hat{p}_j$. We assume the market clears with the maximum competitive prices, as it reflects a model where sellers have market power and in effect set prices for trades. That is, $\hat{p}_j=\overline{p}_j$. As we show in appendix~\ref{app:prices-min}, max price is meaningful to discuss welfare loss due to a selfish platform. 

For a seller $j$, and holding other sellers' strategy fixed, let $\p^{\on}$ and $\p^{\off}$ be the competitive prices when seller $j$ joins and not joins the platform, 
respectively. $j$ joins the platform if the utility from joining is more than that from not joining. That is, if $(1-\alpha)p^{\on}_j > p^{\off}_j$. If there is a tie in utility, the seller  may join or not. If the utility for joining is strictly smaller than not joining, the seller does not join. We can now define a Platform Equilibrium for
 pure strategies.
\begin{definition}[Pure Platform Equilibrium]
 Let $\p(P)$ denote the competitive prices formed when a set of sellers, $P$, join the platform, where for sellers who join, their prices are the maximum competitive price for the obtained network structure. The  set $P$ corresponds to  a \textit{Platform Equilibrium} if and only if,
    \begin{itemize}
        \item For each seller $j\in P$,  then $(1-\alpha)p^{\on}_j \geq p^\off_j$, for $\p^{\on} = \p(P)$ and $\p^{\off} = \p(P\setminus\{j\})$, and
        \item For each seller $j\notin P$, then $p^\off_j \geq (1-\alpha)p^\on_j$, for $\p^\off = \p(P)$ and $\p^\on = \p(P \cup \{j\})$.
    \end{itemize}
\end{definition}

In this way, a pure Platform Equilibrium is a Pure nash equilibrium of the game where the sellers can choose whether or not to join the platform. Similarly, we define the larger set of mixed Platform Equilibria.
\begin{definition}[Mixed Platform Equilibrium]
   Let $x_j\in [0,1]$ denote the probability
with    which seller $j$ joins the platform,
and assume that for sellers who join, their prices are the maximum 
competitive prices for the obtained network structure. The vector $\x=(x_1,\ldots, x_m)$ is a {\em Mixed Platform Equilibrium} if and only if, each seller $j$ maximizes their utility by choosing $x_j$ given $\x_{-j}$,
\begin{eqnarray}
    \E_{P\sim \x} \left[u_j(P,\alpha)\right] \geq \max\left(\E_{P\sim \x_{-j}} \left[u_j(P,\alpha)\right], \E_{P\sim \x_{-j}} \left[u_j(P\cup\{j\},\alpha)\right]\right),
\end{eqnarray}
where $u_j(P,\alpha)$ is the utility of seller $j$ when a set of sellers $P$ joins the platform, the platform charges transaction-fee $\alpha$, and the competitive prices are $\p(P)$.
\end{definition}

\subsection{Modeling the Effect of a Selfish Platform on social welfare} \label{sec:prelims-platform-effect}

We now present our model to analyze the effect of a selfish platform on social welfare. A Selfish platform sets a transaction fee $\alpha$ to maximize its revenue, assuming as before that the market clears with the maximum competitive prices. As there can be a plethora of Nash equilibria, let $\eq(S,B,\mathbf{v},G,\alpha)$ denote the set of equilibria defined by parameters $S,B,\mathbf{v}, G$, at transaction-fee $\alpha$. Further we assume the platform can use its market power to induce the equilibrium with highest revenue in $\eq(S,B,\mathbf{v},G,\alpha)$. The platform chooses $\alpha^{\star}$ such that
\begin{eqnarray}
    \alpha^{\star} \in \argmax_{\alpha}\max_{\x\, \in\, \eq(S,B,\mathbf{v},G,\alpha)}\E_{P\sim \x}\left[\alpha\cdot\sum_{j\in P} \p_j(P)\right]. \label{eq:platform_opt}
\end{eqnarray}

This is the transaction fee that maximizes the platform's revenue for the best possible Nash equilibrium in terms of revenue to the platform. If the platform is regulated in setting fees, and cannot set a fee higher than $\overline{\alpha}$, we simply add the constraint of $\alpha\le \overline{\alpha}$ to the $\argmax$ expression in the platform's optimization problem in Equation~\eqref{eq:platform_opt}.
In this way, the platform acts as a leader that optimizes for an optimistic Stackelberg Equilibrium in a Stackelberg game where both the platform and the sellers act strategically. We denote the set of Platform Equilibria that maximizes the platform's revenue as
\begin{eqnarray}
    \mathbf{X}^{\star} = \argmax_{\x\,\in\, \eq(S,B,\mathbf{v},G,\alpha^{\star})} \E_{P\sim\x}\left[\alpha^{\star}\cdot \sum_{j\in P}\p_j(P)\right].
\end{eqnarray}

We aim to quantify the efficiency of the resulting Platform Equilibrium. Let $k_{B,S}$ be the 
{\em complete bipartite graph between buyers and sellers},
 and let $W^{\star}(S,B,\mathbf{v})=W(S,B,\mathbf{v}, k_{B,S})$ be the \textit{optimal welfare}; i.e., the optimal social welfare in the case that there are no logistical constraints
and each buyer can  transact with any  seller. When clear from context, we denote the optimal welfare $W^{\star}(S,B,\mathbf{v})$ by $W^\star$. 
We use the  {\em price of anarchy}  to quantify the worst-case social welfare in a Platform Equilibrium. 
\begin{definition}[Price of Anarchy]\label{def:poa}
    The {\em price of anarchy} (PoA) of a Platform Equilibrium
is the worst case ratio between the optimal welfare to the lowest-welfare platform equilibrium that maximizes the platform's revenue for any market. Here the worst case is over all possible buyer-seller networks, and buyers' valuations.
\begin{eqnarray}
    PoA  = \max_{S,B,G,\mathbf{v}, \x\in \mathbf{X}^{\star}}\frac{W^\star(S,B,\mathbf{v},G)}{\E_{P\sim \x}\left[W(S,B,\mathbf{v}, G(P)\right]}.
\end{eqnarray}
\end{definition}
Though PoA is defined for all equilibria, in later sections we will slightly abuse notation and use "PoA of pure (mixed) Platform Equilibrium" to denote the the worst case ration between the ideal welfare to the lowest-welfare \textit{pure} (or \textit{mixed}) Platform Equilibrium respectively.

\section{Existence of a Pure Platform Equilibrium}\label{sec:existence_of_pure}
In this section, we show first show that a pure Platform Equilibrium does not always exist in general buyer-seller markets, and then give a polynomial-time algorithm that find a pure Platform Equilibrium in homogeneous-goods market regardless of transaction fee $\alpha$. Most proofs are deferred to Appendix~\ref{app:existence_of_pure}.

\subsection{Non-Existence of Pure Equilibria}

We first show that for some markets, pure Platform equilibria need not exist. 

\begin{restatable}{proposition}{propNoPure}
\label{prop:no_pure}
    There exists a general buyer-seller network and a transaction fee, $\alpha\in [0,1]$, for which there is no pure Platform Equilibrium. 
\end{restatable}

The proof follows the example in Figure~\ref{fig:no_pure_equilibrium}, and 
is given in Appendix~\ref{app:app_no_pure_eq}.
At a high level, a seller's competitive price is related to the externality they impose on the whole market, and is influenced by other sellers' connections. In the proof, we show that sellers' best-response dynamics can exhibit a cyclic behavior: some sellers joining the platform can cause other sellers already on platform to leave, or cause other sellers to join the platform. 

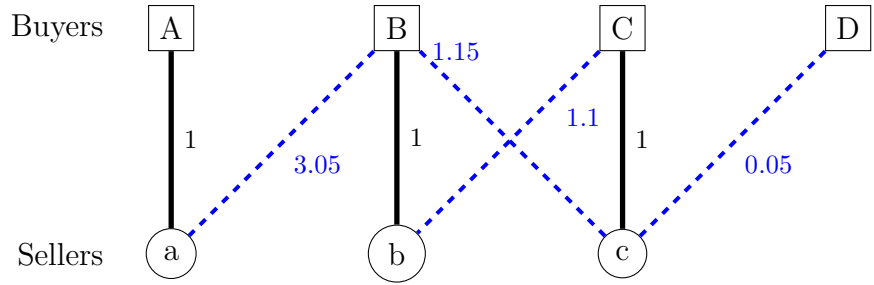
\begin{figure} 
    \centering
    \begin{tikzpicture}[scale=1.5]
        % Square vertices
        \foreach \i/\label in {1/A, 3/B, 5/C, 7/D}
            \node[draw, shape=rectangle, minimum size=0.6cm] (\label) at (\i, 2) {\label};
        
        % Circle vertices
        \foreach \i/\label in {1/a, 3/b, 5/c}
            \node[draw, shape=circle, minimum size=0.6cm] (\label) at (\i, 0) {\label};
        
        % Solid edges with weights
        \foreach \x/\y/\w in {A/a/1, B/b/1, C/c/1}
            \draw[line width=2pt] (\x) -- node[midway, right, font=\footnotesize] {\w} (\y);
        
        % Dashed blue edges with weights
        \foreach \x/\y/\w/\pos/\loc in {B/a/3.05/midway/below right, B/c/1.15/at start/right, C/b/1.1/near start/below right, D/c/0.05/midway/below right}
            \draw[line width=1.5pt, blue, dashed] (\x) -- node[\pos, \loc, font=\footnotesize] {\w} (\y);

        % Buyers caption
        \node[left] at (0.5, 2)  {Buyers};

        \node[left] at (0.5, 0)  {Sellers};
    \end{tikzpicture}
    \caption{An example with no Platform Equilibrium in pure strategies at $\alpha=\frac{1}{2}$. Black solid lines are direct links, as captured by $N(i)$ for buyer $i$, and blue dotted lines indicate missing links. Buyer values are  annotated adjacent to each edge, and all any value that is omitted is zero.} \label{fig:no_pure_equilibrium}
\end{figure} 

\subsection{Computing a Pure Nash Equilibrium in a Homogeneous-Good Market}

We now study homogeneous-good markets, where sellers have identical items, and buyer valuation is captured by a single scalar $v_i\in \mathbb{R}^+$. For this case, a pure equilibrium can always be found in polynomial time. To show this, in Lemma~\ref{lem:optimal_welfare} we first decompose social welfare into optimal welfare achieved by sellers off-platform, and the remaining on platform component. Then in Lemma~\ref{lem:same_on_price} we show all on platform sellers have the same price, and is related to the number of sellers on platform. Finally, we give the algorithm that finds the pure equilibrium. As we show in Theorems~\ref{thm:PE_pure} and~\ref{thm:equilibrium_set_enlarges_1}, in homogeneous goods markets, there is an order of adding sellers to the platform with the following desirable property: when adding a new seller to the platform, (i) 
no seller that is previously added would want to leave, and (ii) a seller that had negative gain from joining before, would rather still not join. All lemmas and theorems in this subsection hold \textit{only} for homogeneous-goods markets. Missing proofs are deferred to Appendix~\ref{app:existence}.

\citet{kranton2000competition} proves when a new seller is added to a market, at most one new buyer transacts (Lemma~\ref{lem:add_one_link}). Building on this we state two useful lemmas on how a set of sellers joining the platform affects the welfare and what is their price. 

\begin{restatable}{lemma}{OptimalWelfare}
\label{lem:optimal_welfare}
    Consider a buyer-seller network $(S,B,G)$. Let $P\subseteq S$ be the sellers joining the platform. Let $B^G$  be buyers transacting in $W(S\setminus P, B,G)$ and $\bar{B}^{G}=B\setminus B^{G}$. Let $m'=|P|$ and let $\bar{v}(m')$ be the sum of the $m'$ largest buyers in $\bar{B}^G$, or the sum of all buyers in $\bar{B}^G$ if there are less then $m'$ buyers in $\bar{B}^G$. Then 
   \begin{eqnarray}\label{eq:optimal_welfare}
        W(S,B,G(P)) = W(S\setminus P, B^G, G) + \bar{v}(m').
    \end{eqnarray}
\end{restatable}

\begin{restatable}{lemma}{SameOnPrice}
\label{lem:same_on_price}
    All sellers joining the platform have the same price of $\bar{v}(m')-\bar{v}(m'-1).$
\end{restatable}

We now present the main result of this section. For any transaction fee $\alpha$ in a homogeneous-goods market, Algorithm~\ref{alg:PE_pure} finds a pure Platform Equilibrium in polynomial time. The algorithm adds sellers to $P$ iteratively, which eventually contains all sellers joining the platform in the pure Platform Equilibirum computed. At every round, the algorithm computes for all sellers who haven't joined $P$ yet, what is their utility gain in joining the platform, given the current $P$ (captured by $\phi_j^P$). The set of sellers with the largest utility gain is labeled as $\Phi_{\max}^P$ in step~\ref{alg:step}. If $|\Phi_{\max}^P|=1$, the algorithm adds the seller with the largest gain to $P$. If $|\Phi_{\max}^P|>1$, it chooses a seller with the lowest off-platform price among all in $\Phi_{\max}^P$. The following lemma shows the seller added to $P$ each round has the lowest off platform price out of all agents in $S\setminus P$.

\begin{algorithm}[h!]
    \begin{enumerate}
        \item Initialize $P\leftarrow \emptyset$.
        \item Let $G(P)$ denote the graph $G$ after connecting all sellers in $P$ to all buyers.  
        \item For each $j\in S$, let $p_j^\on(P)$ be $j$'s price on platform, $p_j^\off(P)$ be $j$'s price off platform given $P\setminus\{j\}$ is on platform.%Notice that for $j\notin P$, $P\setminus\{j\}=P$.
        \item For each $j\in S$, let $\phi_j^P = (1-\alpha) p_j^\on(P) - p_j^\off(P)$. %If $j$ is matched to a buyer through world edges on platform, set $\phi_j^P$ to $-\infty$
        \item Let $\Phi_{\max}^P=\argmax_{j\in S\setminus P}\{\phi_j^P\}$, 
        %be the set of sellers with max gain. Let 
        $\hat{j}\in\argmin_{j\in \Phi_{\max}^P}\{p_{j}^{\off}(P)\}$.  \label{alg:step}
        \item While $(P\neq S) \ \wedge \ (\phi_{\hat{j}}^P \geq 0)$:
         \begin{itemize}
            \item $P\gets P\cup \{\hat{j}\}$.
            \item Update $G(P)$.%, $\bar{B}^P$.
            \item Recompute seller $\hat{j}$ according to the new prices $p_j^\on(P)$ and $p_j^\off(P)$ and the new $\phi_j^P$.
        \end{itemize}
        \item output $P$
    \end{enumerate}
\caption{Finding a pure Platform Equilibrium at fixed $\alpha$.}
\label{alg:PE_pure}
\end{algorithm}

\begin{restatable}{lemma}{AlgoSelect}
\label{lem:algo_select}
    At an iteration, if Algorithm~\ref{alg:PE_pure} adds seller $\hat j$ to platform, then for every seller $j \notin P$, $p_{\hat{j}}^{\off}(P)\leq p_{{j}}^{\off}(P)$. 
\end{restatable}

Using the above lemmas, we now provide a sketch of the proof of the main theorem of the section, in order to demonstrate the main ideas used to proof it. The full proof is quite technical and appears in Appendix~\ref{app:existence}. 

\begin{restatable}{theorem}{PEpure}
\label{thm:PE_pure}
    For any $\alpha\in [0,1]$, Algorithm~\ref{alg:PE_pure} outputs a pure strategy Platform Equilibrium.
\end{restatable}

\begin{proof}[Proof Sketch.]
    The gist of the proof is showing whenever a new seller $\hat{j}$ is added to $P$, the sellers $j$ which are already in  $P$ do not want to leave $P$ (that is, their utility is still higher on-platform). In particular, we want to prove $\forall j\in P$, their benefit of staying on platform is larger than $\hat{j}$'s benefit of joining. That is, given $P$ and $\hat{j}\in \argmax_{j\in S\setminus P}\{\phi_j^P\}$,
    \begin{eqnarray}
        \forall j\in P\quad \phi_j^{P\cup\{\hat{j}\}}\geq \phi_{\hat{j}}^{P}\geq 0.\label{eq:pure_eq_invariant_main_section}
    \end{eqnarray}
    Since $\hat{j}$ is the next seller joining, $\phi_{\hat{j}}^{P}\geq 0$. Expanding Eq.~\eqref{eq:pure_eq_invariant_main_section}, we need to prove the following,
    \begin{eqnarray}
        \phi_j^{P\cup\{\hat{j}\}} - \phi_{\hat{j}}^{P} &= (1-\alpha)[p_j^{\on}(P\cup\{\hat{j}\})-p_{\hat{j}}^{\on}(P)]+[p_{\hat{j}}^{\off}(P)-p_{j}^{\off}(P\cup\{\hat{j}\})]\ge 0.\label{eq:pure_eq_invariant_2_main_section}
    \end{eqnarray}
    Since $j\in P$, both $p_j^{\on}(P\cup\{\hat{j}\})$ and $p_{\hat{j}}^{\on}(P)$ consider the price of $j$ and $\hat{j}$ when the set of sellers on platform is $P\cup \{\hat{j}\}$. Therefore, according to Lemma~\ref{lem:same_on_price}, their price is the same, and the first term in Lemma~\ref{lem:same_on_price} is equal to 0.
    Thus, what's left to prove is 
    \begin{eqnarray}
        p_{\hat{j}}^{\off}(P) \geq p_{j}^{\off}(P\cup\{\hat{j}\}).\label{eq:second-term}    
    \end{eqnarray}
     
    The intuition for this is the following. By Lemma~\ref{lem:algo_select} $j$ has a lower off-platform price than $\hat{j}$ when it is selected by the algorithm. If, as opposed to Eq.~\eqref{eq:second-term}, we have  $p_{j}^{\off}(P\cup\{\hat{j}\}) >p_{\hat{j}}^{\off}(P)$, we show there must be a seller $j'$ joining after $j$, such that its joining the platform frees up a high value buyer for $j$. This implies $j'$ should have high off platform price when it joins, and should be added after $\hat{j}$, forming a contradiction. The full proof also uses a submodularity condition on max weight matching, and is deferred to the Appendix~\ref{app:existence}.
    
\end{proof}

\subsubsection{Useful Properties of Algorithm~\ref{alg:PE_pure}}

In this section, we state other useful properties of pure Equlibria in homogeneous goods markets that arise from Algorithm~\ref{alg:PE_pure}. In Section~\ref{sec:poa_homo_n_ub}, these properties will be useful in lower bounding revenue, and in turn welfare for markets with selfish platforms. 

Algorithm~\ref{alg:PE_pure} can be used to find a pure equilibrium at any $\alpha$. The following two lemmas allows one to continuously decrease transaction fee from $\alpha=1$ with algorithm~\ref{alg:PE_pure}, and gradually enlarge the set of on-platform sellers in equilibrium.

\begin{restatable}{lemma}{PEpureAlpha}\label{lem:pe_pure_alpha1}
    In a homogeneous-goods market $(S,B,G)$, let $S^G$ be the set of transacting sellers, $\bar{S}^G$ be the set of none-transacting sellers. Then any $P\subseteq \bar{S}^G$ joining is a pure equilibrium at $\alpha=1$.
\end{restatable}

\begin{restatable}{lemma}{PEpureContinuous}\label{lem:pe_pure_continuous}
    Let $P_1$ be the set of sellers that forms an equilibrium by Algorithm~\ref{alg:PE_pure} for $\alpha_1$. Then for any $P_2\supsetneq P_1$ there exists an $\alpha_2<\alpha_1$ for which Algorithm~\ref{alg:PE_pure} outputs $P_2$.    
\end{restatable}

Algorithm~\ref{alg:PE_pure}, Lemmata~\ref{lem:pe_pure_alpha1} and \ref{lem:pe_pure_continuous} offer a way to find sets of pure equilibrium at different transaction fees without cold starting the algorithm at $P=\emptyset$ for each $\alpha$. We can continuously decrease the transaction fee from $\alpha=1$ to $0$, adding a seller $j$ to platform when the on-platform gain $\phi_j(P)$ reaches zero. The following theorem further says as $\alpha$ decreases, the set of on-platform sellers in pure equilibrium enlarges by one seller each time. With this procedure, we can check and compare revenue each time a new seller joins. 

\begin{restatable}{theorem}{EnlargeEq}
\label{thm:equilibrium_set_enlarges_1}
    For any number $m_p= 1,...,m$, there exists a pure Platform Equilibrium $P$ such that $|P|=m_p$. Furthermore, all such equilibria can be found by lowering $\alpha$ from $1$ to $0$ with Algorithm~\ref{alg:PE_pure}.
\end{restatable}
\begin{proof}[Proof Sketch.]
    As $\alpha$ decreases continuously from $1$ to $0$, $\phi_j^P$ continuously increases for all $j$ when the equilibrium set $P$ stays the same. When some sellers $\Phi_{\max}^P$ reaches zero on-platform gain, Algorithm~\ref{alg:PE_pure} adds a seller $j\in \Phi_{\max}^P$ to $P$. It suffices to show after $j$ is added to $P$, $\forall \hat{j}\notin P$, either
    \begin{eqnarray*}
    \phi_{\hat{j}}(P\cup \{j\}) & \leq & \phi_{\hat{j}}(P) \text{\; or \;} \phi_{\hat{j}}(P\cup\{j\})<0.
    \end{eqnarray*}
    Therefore, the on platform gain is non-positive for sellers not on platform. Thus, $P\cup \{j\}$ is indeed a pure equilibrium, larger than $P$ by 1.
\end{proof}

Recall that in the example in Proposition~\ref{prop:no_pure}, one seller joining the platform can incentize off-platform sellers to join the platform, or on-platform sellers to leave the platform. This is not true for homogeneous-goods markets. Algorithm~\ref{alg:PE_pure} provides a way to build the set of pure equilibrium. Theorem~\ref{thm:PE_pure} shows that a seller joining would not incentivize on-platform sellers to leave, while Theorem~\ref{thm:equilibrium_set_enlarges_1} proves a seller joining would not incentivize off-platform sellers to join. In Section~\ref{sec:poa_homo_n_ub}, we will use these properties to prove the first results on the positive effect of a selfish platform on the social welfare of the market.

\section{Market Efficiency with an Unregulated Platform}\label{sec:poa_rev_max}

In this section, we study the social welfare with an unregulated platform. Recall our definition of the platform's optimization problem and the definition of the price of anarchy in measuring market efficiency, as in Section~\ref{sec:prelims-platform-effect}. In Section~\ref{sec:poa_homo_n_ub}, we present the first result in the positive effect of a selfish platform on market efficiency. We show that in any homogeneous goods market, the obtained Platform Equilibrium, be it pure or mixed, guarantees at least $1/\log(\min\{m,n\})$-fraction of the ideal social welfare. In Section~\ref{sec:poa_homo_n_lb}, we show that this is tight, up to constant factors. In Appendix~\ref{app:poa_mn_general} we show that in markets with general unit-demand valuations, the state of affairs is worse--the price of anarchy can be as bad as $\min\{n,m\}$. This motivates us in showing that even light regulation of the platform can lead to substantial welfare gain, as presented in Section~\ref{sec:poa-regulated}.

\subsection{$O(\log(\min(n,m)))$ PoA for Homogeneous-Goods Markets}\label{sec:poa_homo_n_ub}

We now show the upper bound of the price of anarchy. 
\begin{theorem}\label{thm:poa_upper_bound_homo}
    The PoA is $O(\log (\min\{n,m\}))$.
\end{theorem}
The proof follows immediately from Lemma~\ref{lem:sppoa_n_leq_m} and~\ref{lem:sspoa_n_geq_m}, which prove this claim for the range $n\leq m$ and for the range $n \geq m$. Interestingly, the proofs give a lower bound on the welfare by giving a lower bound on the platform's optimal revenue, which we denote by $Rev^\star$ in terms of the optimal welfare. As the platform's revenue is not more than the sellers' competitive prices which are no more than buyers values, this immediately gives a lower bound on the social welfare. As we lower bound welfare through revenue, the established bound works for mixed equilibrium.

For the proof of the two lemmas, we use the following notation. Let $B^G$ and $S^G$ be the set of buyers and sellers that transact when there's no platform $(S,B,G)$. Let $\bar{B}^G=B\setminus B^G$ and $\bar{S}^G = S\setminus S^G$. We normalize the ideal welfare to be $W^\star=1$. Denote the $k$-th harmonic number by $H_k=1+\frac{1}{2}+\frac{1}{3}+\ldots+\frac{1}{k}$.

\begin{lemma}
    When there are more sellers than buyers $n\leq m$, the PoA is $O(\log n)$. \label{lem:sppoa_n_leq_m}
\end{lemma}
\begin{proof}
    If $W(S^G,B^G,G) \ge 1/\log n$ we are done. Therefore, assume $W(S^G,B^G,G) < 1/\log n\leq 1/2$, thus $W^\star(S,\bar{B}^G,G) > 1/2$. Let $m^G=|S^G|=|B^G|$, $\bar{m}^G=|\bar{S}^G|$ and $\bar{n}^G= |\bar{B}^G|$ (thus $m=m^G+\bar{m}^G$). Notice that $\bar{m}^G \geq \bar{n}^G$. Sort and denote the values of the buyers in $\bar{B}^G$ by $v_1\geq v_2\geq ... \geq v_{\bar{n}^{G}}$. 

    By Lemma~\ref{lem:pe_pure_alpha1}, the platform can always post $\alpha=1$ and select any set of sellers in $\bar{S}^G$ to join the platform. By Lemma~\ref{lem:optimal_welfare}, if the platform chooses $k=1,\ldots ,\bar{n}^G$ sellers from $\bar{S}^G$ to join, then matching them to buyers $v_1,v_2,...,v_k$ maximizes the welfare. Lemma~\ref{lem:same_on_price} then says all on platform sellers have price $v_k$. Therefore $Rev^{\star}\geq k\cdot v_k$ for any $k\leq \bar{n}^{G}$. As $\bar{m}^G\geq \bar{n}^G, W^{\star}(\bar{S}^G,\bar{B}^G,G) =  W^{\star}(S,\bar{B}^G,G)>\frac{1}{2}$. We get that
    $$ \frac{1}{2}<W^{\star}(\bar{S}^G,\bar{B}^G,G)=\sum_{i=1}^{\bar{n}^G}v_i \leq \sum_{i=1}^{\bar{n}^G}Rev^{\star}/i \leq H_n \cdot Rev^{\star}$$
    which implies $Rev^{\star}=\Omega(1/\log n)$. As the platform's revenue is always smaller than the market's welfare, we get that the welfare is also $\Omega(1/\log n)$.
\end{proof}

\begin{lemma}
    When there are more buyers than sellers $n> m$, the PoA is $O(\log m)$. \label{lem:sspoa_n_geq_m}
\end{lemma}
\begin{proof}
    We follow the same notations as in the proof of Lemma~\ref{lem:sppoa_n_leq_m}. If $W(S^G,B^G,G) \ge 1/\log m$ we are done. Therefore, assume $W(S^G,B^G,G) < 1/\log m$, and thus $W^\star(S,\bar{B}^G,G) > 1-1/\log m$. Denote the values of the buyers in $B^G$ by $\hat{v}_1\geq \hat{v}_2\geq \ldots \geq \hat{v}_{m^G}$ and the values of buyers in $\bar{B}^G$ by $v_1\geq v_2 \geq \ldots \geq v_{n-m^G}$. We consider the following cases:

    \vspace{0.1cm}

    \noindent\textbf{Case 1:} $m^G < \bar{m}^G.$  The reasoning is same with proof in Lemma~\ref{lem:sppoa_n_leq_m}. The platform can always post $\alpha=1$ and select any set of sellers in $\bar{S}^G$ to join the platform. If the platform chooses $k=1,\ldots ,\bar{n}^G$ sellers from $\bar{S}^G$ to join, then on platform sellers have price $v_k$. Therefore $Rev^\star\geq k\cdot v_k$ for any $k\leq \bar{m}^G$. As $m^G < \bar{m}^G$, $W^\star(\bar{S}^G,\bar{B}^G,G)\geq  W^\star(S,\bar{B}^G,G)/2 > (1-1/\log m)/2\geq 1/3$. We get that $$1/3 < W^{\star}(\bar{S}^G,\bar{B}^G,G) =  \sum_{i=1}^{\bar{m}^G} v_i \leq \sum_{i=1}^{\bar{m}^G}Rev^\star/i \leq H_m\cdot Rev^\star,$$ which implies $Rev^\star=\Omega(1/\log m)$ and the welfare is also  $\Omega(1/\log m)$.

    \vspace{0.1cm}

    \noindent\textbf{Case 2:} $m^G \geq \bar{m}^G$ and $v_1\geq 1/12.$ If $\bar{m}^G>0$, then posting $\alpha=1$ guarantees a seller from $\bar{S}^G$ joining and paying $v_1\geq 1/12$ to the platform, which implies that the welfare is also at least $1/12$. Otherwise, by Theorem~\ref{thm:equilibrium_set_enlarges_1} decrease $\alpha$ until one seller $j\in S^G$ joins in pure equilibrium.
    
    As $j$'s off platform price is at most the value of a buyer in $B^G, p^{\off}_j\leq 1/\log m$ before joining and at least $p^{\on}_j\geq v_1-1/\log m\geq 1/24$ after joining. To have $j$ join, $\alpha$ is the highest value for $j$ to break even: $(1-\alpha)p^{on}_j = p^{\off}_j\leq  1/\log m\ \Rightarrow\ \alpha\geq 1-24/\log m \geq 1/2.$ Thus, $Rev^\star \geq \alpha/24 \geq 1/48$, which again implies the welfare at the equilibrium is a constant-fraction of the optimal welfare. 
    
    \vspace{0.1cm}

    \noindent\textbf{Case 3:} $m^G \geq \bar{m}^G$ and $v_1< 1/12.$ Since there are many sellers transacting off-platform, in optimal matching $S^G$ and $\bar{B}^G$ create a lot of welfare.
    \begin{eqnarray*}
        W^\star(S^G,\bar{B}^G,G) &>& W^\star(S^G,B,G)-W^\star(S^G,B^G,G)\\
        &>& \frac{1}{2}W^\star(S,B,G)-W(S^G,B^G,G) > \frac{1}{2}-1/\log m>\frac{1}{3}
    \end{eqnarray*}
    Now let $\ell = \argmax_{\hat{\ell}} \{\sum_{i=1}^{\hat{\ell}} v_i \mbox{ s.t. } \sum_{i=1}^{\hat{\ell}} v_i < 1/6.\}$ that is, the largest index of buyer in $\bar{B}^G$ such that the sum of the $\ell$ highest buyers values in $\bar{B}^G$ is smaller than $1/6$. As $v_1\leq 1/12$, we have $\sum_{i=1}^{\ell} v_i \geq 1/12$. Further since $l$ largest buyers have smaller than $1/6$ welfare but $\bar{B}^G$ provides at least $1/3$ welfare, there are more than $2\ell$ buyers in $\bar{B}^G$. This in turn requires
    \begin{eqnarray}
        \ell\leq m^G/2 \label{eq:ell_ub}
    \end{eqnarray} because otherwise $2\ell>m^G$, sellers in $S^G$ match to no more than $2\ell$ largest buyers in $\bar{B}^G$, creating welfare smaller than $\sum_{i=1}^{2\ell} v_i< 1/6 *2=1/3$. This contradicts with $W^{\star}(S^G,\bar{B}^{G},G)> 1/3$.
    
    As a next step, we lower bound platform's optimal revenue. Theorem~\ref{thm:equilibrium_set_enlarges_1} says as we lower $\alpha$ with Algorithm~\ref{alg:PE_pure}, for every $k\in \{1,2,...,\ell\}$, there is a pure equilibrium with $k$ sellers on platform. Denote the highest transaction fee when $k$ sellers join by $\alpha_k$. By Lemma~\ref{lem:same_on_price}, if $k$ sellers are on platform, they each has a price of the $k$-th largest buyer not transact off-platform, which is weakly larger than $v_k$. Then platform's optimal revenue satisfies $Rev^\star \ge \alpha_k \cdot k \cdot v_k \geq \alpha_l \cdot k \cdot v_k$ for every $k\in\{1,2,...,\ell\}$. We next show that $\alpha_l \geq 1/2$, which implies $$1/12\leq \sum_{i=1}^\ell v_i\leq \frac{1}{\alpha_l}Rev^\star \sum_{k=1}^\ell\frac{1}{k}\leq 2H_\ell Rev^\star \leq 2H_m Rev^\star,$$ and $Rev^\star = \Omega(1/\log m)$, meaning that we get our welfare guarantee.

    We now show $\alpha_{\ell}\geq 1/2$. By Theorem~\ref{thm:equilibrium_set_enlarges_1}, consider when the Algorithm decrease $\alpha$ to $\alpha_{\ell}$ and finally the $l$-th seller $j_{\ell}$ joins the platform in equilibrium. At $\alpha_{\ell}$, the on-platform gain of $j_{\ell}$ just reaches zero: $(1-\alpha_{\ell})p_{j_{\ell}}^{\on}(P) = p_{j_{\ell}}^{\off}(P)$. We lower bound $p_{j_{\ell}}^{\on}(P)$ and upper bound $p_{j_{\ell}}^{\off}(P)$. By Lemma~\ref{lem:same_on_price}, on platform price equals to the valuation of the $\ell$-th largest buyer not transacting off-platform, which is weakly larger than $v_\ell$. If $\ell\leq \bar{m}^G$, by Lemma~\ref{lem:pe_pure_alpha1} all $\ell$ on platform sellers belong to $\bar{m}^G$ and $\alpha_{\ell}=1$. When $j_{\ell}$ considers joining, there are in total $m^G-(\ell-1)$ sellers not joining, and one of them is matched to a buyer with valuation no larger than $\hat{v}_{m^G-\ell+1}$, as each seller from $S^G$ can unmatch at most one buyer in $B^G$. Since the algorithm adds seller $j_\ell$, by Lemma~\ref{lem:algo_select} $j_\ell$ has the lowest off platform price and is at most $p^{\off}_{j_{\ell}}\leq \hat{v}_{m^G-\ell+1}$. Then
    \begin{eqnarray*}
        (1-\alpha_k)v_\ell \leq (1-\alpha_k)p_{j_{\ell}}^\on(P) = p_{j_k}^\off(P) \leq \hat{v}_{m^G-\ell+1}. 
    \end{eqnarray*}
    Assume towards a contradiction that $\alpha_k< 1/2$, we get that 
    \begin{eqnarray}
        \hat{v}_{m^G/2}\geq \hat{v}_{m^G-\ell} \geq \hat{v}_{m^G-\ell+1}> v_\ell/2,
    \end{eqnarray}
    where the first inequality follows Eq.~\eqref{eq:ell_ub}. Thus, we have that
    \begin{eqnarray}
        W(S^G,B^G,G)> \frac{m^G}{2}\cdot\frac{v_\ell}{2} \geq \frac{m\cdot v_\ell}{8},  \label{eq:lb_off_platform_welfare}    
    \end{eqnarray}
    where the second inequality follows $m^G\geq m/2$. On the other hand, we know that 
    \begin{eqnarray}
        m\cdot v_\ell \geq \sum_{i=\ell+1}^{m^G} v_i = \sum_{1}^{m^G} v_i - \sum_{i=1}^{\ell} v_i = W^\star(S^G,\bar{B}^G,G)-\sum_{i=1}^{\ell} v_i \geq 1/3-1/6 = 1/6. \label{eq:up_on_platform_welfare}
    \end{eqnarray}

    Combining Equations~\eqref{eq:lb_off_platform_welfare} and~\eqref{eq:up_on_platform_welfare} yields that  $$W(S^G,B^G,G)\geq 1/48,$$ contradicting $W(S^G,B^G,G)<1/\log m$. Thus, $\alpha_k\geq 1/2$ which concludes the proof.    
\end{proof}

\subsection{$\Omega(\log(\min(n,m)))$ PoA for Homogeneous-Goods Markets}\label{sec:poa_homo_n_lb}

In this section, we give an example market to lower bound PoA $\Omega(\log(m,n))$, showing the tightness of our analysis in Section~\ref{sec:poa_homo_n_ub}.

\begin{theorem}\label{thm:poa_lower_bound_homo}
    There exists a market for which the price of anarchy is $H_{\min\{n,m\}}=\Omega(\log(\min\{n,m\}))$.
\end{theorem}
\begin{proof}
        Figure~\ref{fig:logn_poa} depicts a $n$-buyer $n$-seller market where no transaction can take place without the platform. Buyer $b_1$ has value $n+\epsilon$ for any seller,
and buyer $b_i$, $i=2\ldots n$, has value $\frac{n}{i}$ for any seller.
 Since there are no possible  transactions without joining the platform, each seller weakly prefers to join the platform, so $\alpha^{\star}=1$ and there are multiple equilibria. 
As each buyer values the sellers equally, the platform is in effect choosing, through its choice of $\alpha$, how many sellers join the platform, and in turn the number of buyers that transact. 
    
    If only one seller joins the platform, the competitive price is $n+\epsilon$, and any other number $1<k\leq n$ of sellers joining the platform results in the  competitive price of $n/k$ for each seller
 and total platform revenue $n$. 
For this reason, the platform  prefers one seller to join the platform, in order to maximize its revenue. This results in PoA $\frac{(1+\frac{1}{2}+\frac{1}{3}+\dots+\frac{1}{n})n+\epsilon}{n+\epsilon} \approx H_n$ when $\epsilon$ is small.

If we increase the number of sellers $m$ to be larger than $n$, then still the revenue optimal equilibrium is the one where only one seller joins, whereas the welfare-optimal allocation matches $n$ sellers to $n$ buyers, obtaining a  the $H_n= H_{\min\{n,m\}}$ lower bound. If $m$ is smaller than $n$, we can have $m$ buyers with values $1+\epsilon$, $1/2$, $1/3$,\ldots,$1/m$ and then $n-m$ buyers with value 0, and get a $H_m =H_{\min\{n,m\}}$ lower bound.
\end{proof}

\begin{figure} 
    \centering
    \begin{tikzpicture}[scale=1.5]
        \foreach \i/\label in {1/$b_1$, 2.5/$b_2$, 4/$b_3$, 7/$b_i$,10/$b_n$}
            \node[draw, shape=rectangle, minimum size=0.6cm] (\label) at (\i, 2) {\label};
        
        % Circle vertices
        \foreach \i/\label in {1/$s_1$, 2.5/$s_2$, 4/$s_3$, 7/$s_i$,10/$s_n$}
            \node[draw, shape=circle, minimum size=0.6cm] (\label) at (\i, 0) {\label};
            
        % value vertices
        \foreach \i/\j in {1/$n+\epsilon$, 2.5/$\frac{n}{2}$, 4/$\frac{n}{3}$, 7/$\frac{n}{i}$, 10/$1$}
            \node at (\i, 2.5) {\j};
        
        \foreach \i in {$b_1$,$b_2$,$b_3$,$b_i$,$b_n$}
            \foreach \j in {$s_1$,$s_2$,$s_3$,$s_i$,$s_n$}
                \draw[line width=1.1pt, blue, dashed] (\i) -- (\j);
     
        % Ellipsis
        \node at (5.5, 2) {$\ldots$};
        \node at (8.5, 2) {$\ldots$};
        \node at (5.5, 0) {$\ldots$};
        \node at (8.5, 0) {$\ldots$};
        \node at (5.5, 2.5) {$\ldots$};
        \node at (8.5, 2.5) {$\ldots$};

        % captions
        \node[left] at (0.5, 2)  {Buyers};
        \node[left] at (0.5, 0)  {Sellers};
        \node[left] at (0.5, 2.5)  {Values};
    
    \end{tikzpicture}
        \caption{A homogeneous goods market with price of anarchy $H_n$. There are no off-platform edges. Buyer $b_1$ has  value $n+\epsilon$ for any seller, and buyer $b_i$, $i=2,\ldots,n$, has value $n/i$ for any seller.} \label{fig:logn_poa}
\end{figure}
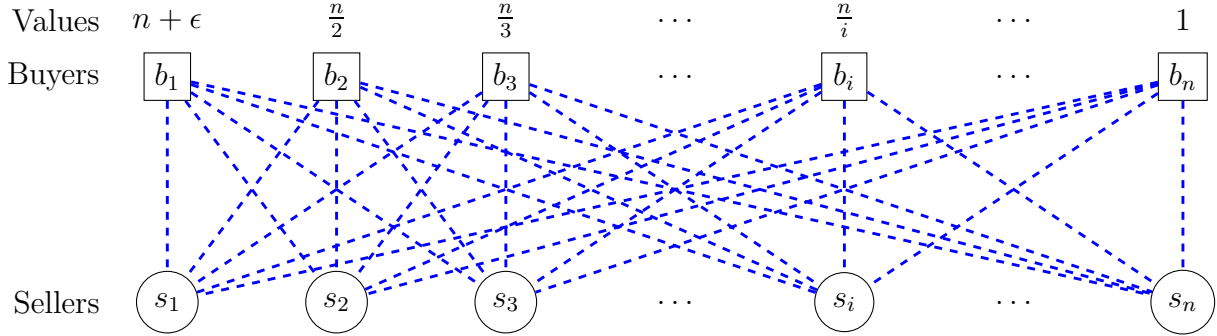

\section{Main Positive Result: Market Efficiency with a Regulated Platform} \label{sec:poa-regulated}

The previous section shows PoA can scale linearly in unregulated platforms. In this section, we show even light regulation can substantially increase welfare in markets with general unit-demand valuations. We prove this by providing a lower bound on the worst possible welfare at a given fee cap $\alpha$. Since the welfare guarantee is monotonically decreasing in $\alpha$, we assume the platform posts the cap $\alpha$ as its fee. It might be the case that in order to maximize its revenue, the revenue sets a lower transaction fee, but this only gives a higher welfare guarantee. Interestingly, the results in this section do not assume the platform can choose a specific Platform Equilibrium, as we prove the guarantee for every fee smaller than the cap, and for every equilibrium at a given fee.

This section is structured as follow. In Section~\ref{sec:poa_cap_pure}, we give a tight bound on the Price of Anarchy of pure Platform equilibria for a fixed transaction fee $\alpha$. As Proposition~\ref{prop:no_pure} shows a a pure equilibrium need not exist, in Section~\ref{sec:poa_mixed} we extend the analysis to the set of mixed Platform Equilibria. In Section~\ref{sec:poa_tight}, we show our analysis is tight for any fee $\alpha$.

\subsection{Pure Equilibrium Bound} \label{sec:poa_cap_pure}

Bounding the PoA of a Platform Equilibrium requires buyers' valuations for the allocation at the equilibrium while the strategic actors in the model are the sellers. In order to relate the two we use the characterization of maximum competitive prices given by Eq.~\eqref{eq:max_price}, as they directly relate to the sellers' utility as well as to buyers' welfare. 
\begin{theorem}
    The Price of Anarchy of pure Platform Equilibrium, when the platform sets a transaction-fee $\alpha\in [0,1)$, is at most $\frac{2-\alpha}{1-\alpha}$. \label{thm:pure_poa}
\end{theorem}
\begin{proof}
    Let $v_1,\ldots, v_n$ be buyers' valuations,
 and let $G$ be a buyer-seller network. We first show that when no seller chooses to join the platform, we already get the welfare guarantee: 
$$\frac{W^\star}{W(S,B,G)} \leq \frac{2-\alpha}{1-\alpha}.$$ 

    Consider this case of $P=\emptyset$, and denote $G(j)$ as the network when only seller $j$ joins the platform. Let 
    \begin{eqnarray}
        p_{j}^{\on} = W(S,B,G(j)) - W(S\setminus\{j\},B,G(j)) = W(S,B,G(j)) - W(S\setminus\{j\},B,G)\label{eq:p_on}
    \end{eqnarray} 
    denote $j$'s price when $j$ joins the platform. Let $p_{j}^{\off}$ denote seller $j$'s price when not joining the platform. We can bound this price by $j$'s maximum competitive price,
    \begin{eqnarray}
        p_{j}^{\off} \leq W(S,B,G) - W(S\setminus\{j\},B,G). \label{eq:p_off}
    \end{eqnarray} 

    As seller $j$ does not join the platform, we have  $$p_{j}^{\off} \geq (1-\alpha)p_{j}^{\on}.$$ By Equations~\eqref{eq:p_on} and~\eqref{eq:p_off}, we have
    \begin{eqnarray*}
        W(S,B,G(j)) - W(S,B,G) & \leq & p^{\on}_j-p^\off_j \\ 
        &\leq & \alpha\cdot p^\on_j \\ 
        & \leq & \frac{\alpha}{1-\alpha}\cdot p^\off_j \\
        & \leq & \frac{\alpha}{1-\alpha}\left(W(S,B,G) - W(S\setminus\{j\},B,G)\right).
    \end{eqnarray*}

    Rearranging gives 
    \begin{eqnarray}
    W(S,B,G(j)) \leq
        \frac{1}{1-\alpha}W(S,B,G) - \frac{\alpha}{1-\alpha}W(S\setminus\{j\},B,G). \label{eq:bound1}
    \end{eqnarray}

    Let $i^\star(j)$ be the buyer matched to $j$ in the ideal matching $W^\star$. Since this is also one of the options to match $j$ in $G(j)$ is to $i^\star(j)$, we have
    \begin{eqnarray}
        W(S,B, G(j)) & \geq& v_{i^\star(j)j} + W(S\setminus\{ j\}, B\setminus\{i^\star(j)\}, G). \label{eq:bound2}
    \end{eqnarray}

    Combining Equations~\eqref{eq:bound1} and~\eqref{eq:bound2}, we get
    \begin{eqnarray}
        v_{i^\star(j)j} \leq \frac{1}{1-\alpha}W(S,B,G)-\frac{\alpha}{1-\alpha}W(S\setminus\{j\},B,G)-W(S\setminus\{j\},B\setminus\{i^\star(j)\},G). \label{eq:bound3}
    \end{eqnarray}

    We wish to relate the right-hand side of Eq.~\eqref{eq:bound3} to terms that relate to $W(S,B,G)$ and $W^\star$.
 Let $i^G(j)$ be the buyer matched to $j$ in $W(S,B,G)$. First, consider $W(S\setminus\{j\},B,G)$. By 
definition, we have
    \begin{eqnarray*}
        W(S,B,G) = W(S\setminus\{j\}, B\setminus\{i^G(j)\}, G) + v_{i^G(j)j}. 
    \end{eqnarray*}
    Thus, we have,
    \begin{eqnarray}
        W(S\setminus\{j\},B,G) \geq W(S\setminus\{j\}, B\setminus\{i^G(j)\}, G) = W(S,B,G)-v_{i^G(j)j}.\label{eq:bound4}
    \end{eqnarray}

    As for $W(S\setminus\{j\},B\setminus\{i^\star(j)\},G)$, 
seller  $j$ is matched to $i^G(j)$ in $G$, while $i^\star(j)$ is matched to some potentially different vertex in $G$, which we call \textit{the twin of $j$}
 and denote by $t(j)$. We  have the following inequality,
    \begin{eqnarray*}
        W(S,B,G) &\leq & W(S\setminus\{j,t(j)\},B\setminus\{i^G(j),i^\star(j)\},G) + v_{i^G(j)j} + v_{i^\star(j)t(j)}\nonumber \\
        &\leq & W(S\setminus\{j\},B\setminus\{i^\star(j)\},G) + v_{i^G(j)j} + v_{i^\star(j)t(j)}, 
    \end{eqnarray*}
    where the first inequality is an equality if $j\neq t(j)$. Rearranging gives
    \begin{eqnarray}
        W(S\setminus\{j\},B\setminus\{i^\star(j)\},G) \geq W(S,B,G) - (v_{i^G(j)j} + v_{i^\star(j)t(j)}).\label{eq:bound5}
    \end{eqnarray}

    Combining Equations~\eqref{eq:bound3},~\eqref{eq:bound4}, and~\eqref{eq:bound5}, we get
    \begin{eqnarray}
        v_{{i^\star}(j)j} &\leq& \frac{1}{1-\alpha}W(S,B,G)-\frac{\alpha}{1-\alpha}\left(W(S,B,G)-v_{i^G(j)j}\right) - \left(W(S,B,G) - (v_{i^G(j)j} + v_{i^\star(j)t(j)})\right) \nonumber \\
        & = & \frac{1}{1-\alpha}\cdot v_{i^G(j)j} + v_{i^\star(j)t(j)}. \label{eq:combine_three_terms_right}
    \end{eqnarray}

    Summing over all sellers $j$, we have
    \begin{eqnarray*}
        W^\star = \sum_j v_{{i^\star}(j)j} & \leq & \sum_j \left(\frac{1}{1-\alpha}\cdot v_{i^G(j)j} + v_{i^\star(j)t(j)}\right)\\ & = &  \frac{1}{1-\alpha}\cdot\sum_j v_{i^G(j)j} + \sum_j v_{i^\star(j)t(j)}\\
        & \leq & \frac{1}{1-\alpha} W(S,B,G) + \sum_j v_{i^G(j)j}\\
        & = & \frac{2-\alpha}{1-\alpha} W(S,B,G),
    \end{eqnarray*}
    where the last  inequality follows because each seller $j$ has a distinct twin $t(j)$, so we sum over distinct edges of $G$.   

    To finish the proof, 
we consider a case where a non-empty
set $P$ of sellers join in a Nash equilibrium. Let $G(P)$ denote the network where sellers in $P$ are connected to all buyers,
 and sellers not in $P$ are connected only via links in the original network $G$.
 Now consider another market where the buyers have the same valuations as the original market, but the  initial network is $G'=G(P)$. For the same transaction fee $\alpha$, it is
 an equilibrium  for no seller to join the platform  in $(S,B,G')$.
 Indeed, sellers in $P$ are already linked to all buyers in this modified market,
 so their utility can only decrease in joining,
 and if there was a seller $j\notin P$ that would strictly prefer to join at $\alpha$, then
 $P$ is not an equilibrium in the original market.
    By the above, we have $\frac{W^\star(S,B,G')}{W(S,B,G')}\leq \frac{2-\alpha}{1-\alpha}$. Clearly, $W^\star(S,B,G')=W^\star(S,B,G)$ and  $W(S,B,G')=W(S,B,G(P))$, which concludes the proof.
\end{proof}

Theorem~\ref{thm:pure_poa} bounds the PoA of pure equilibrium at $\alpha$, should it exist. However,
we know from Proposition~\ref{prop:no_pure} that  there exists markets and transaction fees $\alpha$ for which there is no pure Platform Equilibrium.
 This motivates us to extend our analysis to mixed equilibria.

\subsection{Mixed Equilibrium Bound} \label{sec:poa_mixed}

In this section, we prove the same bound works for mixed equilibria. For mixed equilibria, sellers randomize between joining the platform and not joining (and thus only being able to transact with buyers who whom they have an existing connection). This possibly results in $2^{m}$ graph realizations and a competitive price for each of them. To reason about the probability of each realization through sellers' incentives, and taking expectation over the max welfare of all graphs is quite hard.

We thus reduce the argument for PoA of mixed equilibria to the case of pure equilibria. One way to do this is to reduce randomization of sellers' mixed strategy to randomization of nature's choice, in a \textit{Bayesian game}. Take a mixed Nash equilibrium, $\x$, and define a Bayesian game where sellers have pure strategies (joining vs.~not joining), but where in the event
that a seller $j$ does not join they are still connected to all buyers with 
 probability $x_j$.
We show that one of the \textit{Bayes-Nash Equilibrium} in this Bayesian game has no seller choosing to join, which corresponds to the smallest-welfare equilibrium. We then show this smallest-welfare equilibrium has the same expected social welfare as the mixed Nash equilibrium 
in the corresponding complete information game. The rest of the proof follows the same arguments as the first part of Theorem~\ref{thm:pure_poa}, but where the quantities for each seller $j$ are in expectation over the links formed. In the following, we give a sketch that gives the idea of the construction of the Bayesian game, and defer the full proof to  Appendix~\ref{app:mixed_proof}. 
\begin{restatable}{theorem}{mixedPoA}
\label{thm:mixed_poa}
    The Price of Anarchy of mixed Platform Equilibrium,  when the platform sets a transaction-fee $\alpha\in [0,1)$, is at most $\frac{2-\alpha}{1-\alpha}$.
\end{restatable}
\begin{proof}[Proof Sketch]
    Consider a mixed Platform Equilibrium, $\x = (x_1,\ldots, x_n)$, for a buyer-seller network $(S,B,G)$, where $x_j$ is the probability that seller $j$ joins the platform. We define the following Bayesian game:
    \begin{itemize}
      %  \item $G$ is the base graph.
        \item For each seller $j$, with probability $x_j$, $j$ can transact with all buyers ("type 1"), and with probability $1-x_j$, it can only transact with the buyers linked to $j$ in $G$ ("type 2").
        \item The platform posts a transaction fee $\alpha$.
        \item Before knowing the realization of its type,
 a seller can choose whether or not to join the platform, in which case the seller 
can with certainty transact with all buyers.
        \item Given the realized graph $G'$,which depends on the realized types as well as actions of each seller, a competitive equilibrium with the maximum price is formed. %Market clears with the maximum competitive prices.
        \item A seller joining the platform pays an $\alpha$-fraction of their revenue to the platform, regardless of its type.% \dcp{also those with type 1?}
    \end{itemize}

    We  now show that there exists a smallest-welfare equilibrium in the Bayesian game where no seller joins. %\dcp{`worst-case equilibrium`, not `bad`?}
For this, consider a seller $j$ that adopts probability $x_j$ in the mixed Nash Platform Equilibrium
of the original, complete information platform game. There are three cases to consider:

\noindent\textbf{Case 1:} $x_j=0$ In this case, $j$'s expected utility from not joining the platform \textit{in the original, complete information game} is at least as much as $j$'s expected utility from joining given $\x_{-j}$. \textit{In the Bayesian game}, given that no other seller joins the platform, $j$'s utility from either joining or not joining the platform is exactly $j$'s utility for joining or not joining in the complete information game,
 as other links are formed according to $\x_{-j}$.

\noindent\textbf{Case 2:} $x_j=1$: In this case, in the Bayesian game, $j$ is linked to all buyers with probability $1$. Thus,  joining the platform in the Bayesian game does not increase $j$'s utility.

\noindent\textbf{Case 3:} $x_j\in (0,1)$: In this case, in the original, complete information game, $j$'s expected utility is the same for joining and not joining given $\x_{-j}$; otherwise, $j$ would deviate and $\x$ would not be an equilibrium.
 Let $u_j^\x$ denote $j$'s utility at $\x$ in the original game, 
which is also $j$'s utility from either joining or not joining the platform in the original game. 
        
\textit{In the Bayesian game}, consider the case where $j$ does not join the platform, and when no other seller joins the platform. If $j$'s links stay as they were in $G$, which happens with probability $1-x_j$, then $j$'s expected utility is exactly $u_j^\x$, as this is $j$'s utility from not joining the platform in the original, complete information game when all other sellers' links are formed according to $\x_{-j}$.
 If $j$'s links are formed, which happens with probability $x_j$,  then  $j$'s expected utility is at least as much as when it joins in the original game given all other sellers' links are formed according to $\x_{-j}$. This is because the  distribution on link formation is the same, while the seller does not have to pay the $\alpha$ fee to the platform. 
Therefore, $j$'s expected utility from not joining is at least $u_j^\x$. 

If $j$ does join the platform in the Bayesian game, then $j$'s utility is exactly the utility of $j$ for joining the platform in the original game, as the  distribution on link formation is the same as if they join in the original setup, and the fee $j$ pays is the same. Therefore, $j$'s utility is exactly $u_j^\x$. We get that $j$'s utility for not joining in the Bayesian game is at least as $j$'s utility for joining.

\vspace{0.3cm}

In all the above cases, for each seller $j$,  their expected utility for not joining the platform in the Bayesian game, given all other sellers do not join, is at least as much as  their 
expected utility for joining the platform in the Bayesian game.
Thus,  there is an equilibrium  
 in the Bayesian game where no seller joins.
The expected welfare of this equilibrium is exactly the same as the expected welfare of the original complete information game in mixed Platform Equilibrium $\x$, as we have the same distribution over the formation of links, and the competitive equilibrium formed 
always maximizes the welfare given  the links. 
    
    To conclude the proof, we
 show that if no sellers join in the Bayesian game,
 then the PoA 
with respect to pure strategies in the Bayesian game
is at most $\frac{2-\alpha}{1-\alpha}$. 
The proof follows the same arguments of the proof of this case in Theorem~\ref{thm:pure_poa}, but the quantities we reason about for seller $j$ are in expectation over $\x_{-j}$. For completeness, we give the full analysis in Appendix~\ref{app:mixed_proof}.      
\end{proof}

In markets without a platform, social welfare is defined as the sum of buyers' and sellers' utilities. We notice that the above result also implies welfare guarantees if we consider the benchmark used in settings without a platform. 

\begin{corollary}
    For any Platform Equilibria with $\alpha$ transaction fee, the sum of buyers and sellers' utility is at least $\frac{(1-\alpha)^2}{2-\alpha}$ fraction of the optimal welfare $W^\star$.
\end{corollary}
\begin{proof}
    When the set of sellers $P$ join the platform in equilibrium, platform's revenue is at most $\alpha W(S,B,G(P))$. The sum of buyers and sellers' utility is 
    $$ (1-\alpha)W(S,B,G(P))\geq (1-\alpha)\frac{1-\alpha}{2-\alpha}W^\star$$
\end{proof}

\subsection{Tight Instance} \label{sec:poa_tight}

We now show that our PoA analysis is tight: for any transaction fee $\alpha$, there exists an instance of the platform game in which the price of anarchy of Platform Equilibria with pure strategies is exactly $\frac{2-\alpha}{1-\alpha}$.
\begin{theorem}\label{thm:poa_tight}
    For every transaction fee $\alpha\in[0,1)$,  
    there exists a market for which the ratio of optimal welfare to social welfare of a Platform Equilibrium is indeed $\frac{2-\alpha}{1-\alpha}$, and in that Platform Equilibrium sellers use pure strategies.
\end{theorem}
\begin{proof}
For any given transaction fee $\alpha\in [0,1)$, %\dcp{or $[0,1)$?}, 
Figure~\ref{fig:tight_bound_poa_for_alpha_transaction_fee} gives a three-buyer three-seller market with a tight PoA. At transaction fee $\alpha$, no seller joining the platform is a Nash equilibrium, 
since the off-platform price of $1$ is larger than the seller's on-platform utility,
which is $(1-\alpha)p^{on}=(1-\alpha)\left[\frac{2-\alpha}{1-\alpha}-\epsilon-1\right]=1-\epsilon+\alpha\epsilon$. As $\epsilon$ goes to zero in the parameterized values of buyers in the market,
 the price of anarchy goes to $\frac{2-\alpha}{1-\alpha}$. Note all sellers joining the platform is another Nash equilibrium. But the price of anarchy depends 
on the worst possible Nash equilibrium.
\end{proof}

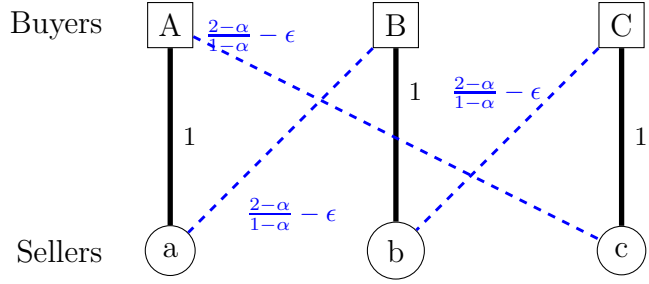
\begin{figure} 
    \centering
    \begin{tikzpicture}[scale=1.5]
        % Square vertices
        \foreach \i/\label in {1/A, 3/B, 5/C}
            \node[draw, shape=rectangle, minimum size=0.6cm] (\label) at (\i, 2) {\label};
        
        % Circle vertices
        \foreach \i/\label in {1/a, 3/b, 5/c}
            \node[draw, shape=circle, minimum size=0.6cm] (\label) at (\i, 0) {\label};
        
        % Solid edges with weights
        \foreach \x/\y/\w/\pos in {A/a/1/midway, B/b/1/near start, C/c/1/midway}
            \draw[line width=2pt] (\x) -- node[\pos, right, font=\footnotesize] {\w} (\y);
        
        % Dashed blue edges with weights
        \foreach \x/\y/\w/\pos/\loc in {B/a/$\frac{2-\alpha}{1-\alpha}-\epsilon$/near end/below right, A/c/$\frac{2-\alpha}{1-\alpha}-\epsilon$/at start/right, C/b/$\frac{2-\alpha}{1-\alpha}-\epsilon$/near start/left}
            \draw[line width=1.1pt, blue, dashed] (\x) -- node[\pos, \loc, font=\footnotesize] {\w} (\y);

        % Buyers caption
        \node[left] at (0.5, 2)  {Buyers};

        \node[left] at (0.5, 0)  {Sellers};
    
    \end{tikzpicture}
    \caption{An example with a tight Bound for PoA for an $\alpha$ transaction fee, for any $\alpha \in [0,1)$. Black solid lines are direct links, as captured by $N(i)$ for buyer $i$, and blue dotted lines indicate missing links. Buyer values are annotated adjacent to each edge, and all any value that is omitted is zero.} \label{fig:tight_bound_poa_for_alpha_transaction_fee}
\end{figure}

\section{Generalizations}
In this section, we extend our previous results to more general settings. We analyze markets with multiple platforms in Section~\ref{sec:multi_platform}, sellers with production costs in Section~\ref{sec:seller_with_cost}, and show all of our results hold when buyers have additive-over partition valuation in Section~\ref{sec:beyond-ud}. Most proofs are delayed until Appendix~\ref{app:generalizations}.

\subsection{Market with Multiple Platforms}\label{sec:multi_platform}

For this section, we prove that our results in Section~\ref{sec:poa-regulated} hold in markets with multiple competing platforms. Consider the buyer-seller network on graph $G$, as defined in Section~\ref{sec:prelims}. There is a set of $R$ platforms $F=\{f_1,f_2,...,f_R\}$, where each platform $f_r$ has a subset of buyers $B_r\subset B$ and transaction fee $\alpha_i \leq \alpha$. A buyer can be on multiple platforms.
The regulator imposes an $\alpha$ cap on transaction-fee for all platforms. A seller $j$ can join multiple platforms $F_j\subseteq F$ and transact with all buyers on the platforms it joins, $\cup _{f_r\in F_j}B_r$, as well as the buyers it knows on $G$. The market clears according to competitive equilibrium. If $j$ transacts through platform edges with some buyer $i$, it pays $\alpha_r p_j^{\on}$ to the platform $f_r$ corresponding to buyer $i$. If $i$ is on multiple platforms that $j$ joins, then $j$ pays to the platform with the lowest fee. Everything else remains the same with Section~\ref{sec:poa-regulated}.

\begin{restatable}{theorem}{MultiPlatform}
    For a market with multiple platforms and transaction-fee cap $\alpha \in [0,1)$, the Price of Anarchy of Platform Equilibrium is at most $\frac{2-\alpha}{1-\alpha}$. \label{thm:pure_poa_multiple_platform}
\end{restatable}
The proof is similar to that of Theorem~\ref{thm:pure_poa} and \ref{thm:mixed_poa}, and is deferred to Appendix~\ref{app:multiple_platforms}.

\subsection{Sellers with Production Costs}\label{sec:seller_with_cost}
In this section, we extend our results to sellers $j\in S=\{1,2,...,m\}$ with production cost $c_j\geq 0$. We only present the main result and defer the readers to Appendix~\ref{app:seller_with_cost} for the full section.

\begin{restatable}{theorem}{purePoACost}
\label{thm:pure_poa_cost}
    In a market with cost, let $\beta=\sum_{j\in S}c_j / W^\star(S,B,\mathbf{v},\mathbf{c},G)$ denote the fraction of cost to optimal welfare. The Price of Anarchy of the Platform Equilibrium,  when the platform sets a transaction-fee $\alpha\in [0,1)$, is at most $\frac{2-\alpha}{1-\alpha-\alpha\beta}$.
    \label{thm:pure_poa_cost}
\end{restatable}

The above theorem implies that at the same transaction fee, the social welfare guarantee deteriorates as costs get higher. This was specially relevant during the COVID-19 crisis, where demand and supply surged on digital platforms \citep{raj2020covid} and the cost of production grew \citep{felix2020us}.
At $30\%$ of transaction fee (Table~\ref{table:platforms and their coommission rate}) and $30\%$ of food cost \footnote{Most profitable restaurants aim for a food cost percentage between 28 and 35\%. This does not include other cost factors such as labor, rentals, etc. Figure is taken in 2023 from Doordash website \url{https://get.doordash.com/en-us/blog/food-cost-percentage}}, the resulting welfare at a  Platform Equilibrium  is at least $23.5\%$ of the ideal welfare. This is smaller than the $41\%$ without costs, further demonstration the need to regulate transaction fee during periods where production costs increase.

\subsection{Additive-over-Partition Buyer Valuation} \label{sec:beyond-ud}
Our results also extend beyond unit demand valuation. For example, consider the case where buyers have additive valuation. In this case, we can use $m$ vertices to represent each buyer in the bipartite graph, where each vertex only values one seller. We show how to generalize our results to more complex valuations, such as buyers valuing at most $c_{\ell}$ items  each group of sellers $S_\ell$ who sell similar items. Thus, the valuation of a buyer is additive constrained by the feasible sets defined by a partition matroid $\mathcal{M}$.

We formally define the class of buyers' valuations, which we term \textit{additive-over partition valuations}, and show results in Section~\ref{sec:poa-regulated} extend.

For each buyer $i\in\{1,2,...,n\}$, there is a partition of sellers $S_i = S_{i,1}\cup S_{i,2}\cup...\cup S_{i,T_i}$ and a  capacity $c_{i,\ell}$ associated with each seller group $S_{i,\ell}$.  Buyer $i$ values seller $j$'s item at $v_{ij}$, and the bundle of items $\alloc$ at
\begin{eqnarray}
    v_i(\alloc) = \sum_{\ell\in [T_i]} \max_{T\subset S_{i,\ell} \ : \ |T|\le c_{i,\ell}} \sum_{j\in T}v_{ij}. 
\end{eqnarray}

We now show how to encode such valuations into our  bipartite graph model. For each buyer $i$ and for each seller group $\ell$, we introduce $c_{i,\ell}$ unit-demand buyers. Each such buyer $i'\in [c_{i,\ell}]$ will have the following unit-demand valuation 
\begin{eqnarray*}
    v'_{i',j} = \begin{cases} v_{i,j} \quad & j\in S_{i,\ell}\\
                           0 \quad & j\notin S_{i,\ell}
            \end{cases},
\end{eqnarray*}
which ensures such buyer transacts with only sellers in $S_{i,\ell}$. For each such buyer $i'$, we also keep the same links to sellers such as the original buyer $i$. Via this reduction, our results in Section~\ref{sec:poa-regulated} trivially hold for additive-over-partition buyers as well.

\section{Conclusion and Future Directions}
In this paper, we initiate the rigorous study of the effect of platforms on the welfare of a constrained market, modeled by a buyer-seller network with missing links. We study the network structure of the market, and give an algorithm that always finds a pure equilibrium in homogeneous-goods markets. We show when the platform is unregulated, and can set fees to selfishly maximize its revenue, then the Price of Anarchy can be as bad as $\min\{m,n\}$ in general markets, and $\Theta(\log(\min\{n,m\} ) )$ even in homogeneous-goods markets. However, if regulation requires the platform to charge no more than $\alpha$-fraction of sellers' revenue, then we get a tight $\frac{2-\alpha}{1-\alpha}$-fraction of the optimal \textit{unconstrained} welfare in any equilibrium.
Moreover, even when discounting the platform's revenue, the surplus left for the buyers and sellers in any equilibrium is  at least $\frac{(1-\alpha)^2}{2-\alpha}$-fraction of the optimal unconstrained welfare. 
Finally, we extend our analysis to multiple platforms, buyer valuations beyond unit demand, and sellers with non-zero production cost.

While this work already provides strong insights into understanding a platform's effect on a market's welfare,  there are many exciting directions to take:
\begin{enumerate}
    \item Although the $\Omega(\log (\min\{n,m\}))$ lower bound on the Price of Anarchy is discouraging, this is a  contrived instance. We believe that for constrained and natural market structures, we can avoid the logarithmic dependence in the number of buyers/sellers. 
    \item While we extended the analysis to additive-over-partition valuations in Appendix~\ref{app:beyond-ud}, there remains the need to incorporate more general classes of valuation functions, such as the commonly studied like gross-substitutes valuations. 
    \item Extending our model to accommodate multi-supply sellers is not straightforward. A seller with multiple items can limit supply to raise revenue. And her price becomes zero if some of her items are left unsold. This requires further extension of the model to explain how and why such sellers join the platform.
    \item While seller-side transaction-fees are common, other fees are used by platforms, such as buyer and seller monthly/annual registrations fees. 
    This changes buyers and sellers' strategic behaviors to join the platform, and will required further extensions of the model.
\end{enumerate}

\bibliographystyle{ACM-Reference-Format}
\bibliography{main}


\begin{thebibliography}{22}

%%% the Web address.

\ifx \showCODEN    \undefined \def \showCODEN     #1{\unskip}     \fi
\ifx \showDOI      \undefined \def \showDOI       #1{#1}\fi
\ifx \showISBNx    \undefined \def \showISBNx     #1{\unskip}     \fi
\ifx \showISBNxiii \undefined \def \showISBNxiii  #1{\unskip}     \fi
\ifx \showISSN     \undefined \def \showISSN      #1{\unskip}     \fi
\ifx \showLCCN     \undefined \def \showLCCN      #1{\unskip}     \fi
\ifx \shownote     \undefined \def \shownote      #1{#1}          \fi
\ifx \showarticletitle \undefined \def \showarticletitle #1{#1}   \fi
\ifx \showURL      \undefined \def \showURL       {\relax}        \fi
% invisible to TeX
\providecommand\bibfield[2]{#2}
\providecommand\bibinfo[2]{#2}
\providecommand\natexlab[1]{#1}
\providecommand\showeprint[2][]{arXiv:#2}

\bibitem[Ahuja et~al\mbox{.}(2021)]%
        {mckinsey2021}
\bibfield{author}{\bibinfo{person}{Kabir Ahuja}, \bibinfo{person}{Vishwa
  Chandra}, \bibinfo{person}{Victoria Lord}, {and} \bibinfo{person}{Curtis
  Peens}.} \bibinfo{year}{2021}\natexlab{}.
\newblock \showarticletitle{Ordering in: The rapid evolution of food delivery}.
\newblock \bibinfo{journal}{\emph{McKinsey \& Company}} (\bibinfo{year}{2021}).
\newblock


\bibitem[Armstrong(2006)]%
        {armstrong2006competition}
\bibfield{author}{\bibinfo{person}{Mark Armstrong}.}
  \bibinfo{year}{2006}\natexlab{}.
\newblock \showarticletitle{Competition in two-sided markets}.
\newblock \bibinfo{journal}{\emph{The RAND journal of economics}}
  \bibinfo{volume}{37}, \bibinfo{number}{3} (\bibinfo{year}{2006}),
  \bibinfo{pages}{668--691}.
\newblock


\bibitem[Banerjee et~al\mbox{.}(2017)]%
        {banerjee2017segmenting}
\bibfield{author}{\bibinfo{person}{Siddhartha Banerjee},
  \bibinfo{person}{Sreenivas Gollapudi}, \bibinfo{person}{Kostas Kollias},
  {and} \bibinfo{person}{Kamesh Munagala}.} \bibinfo{year}{2017}\natexlab{}.
\newblock \showarticletitle{Segmenting two-sided markets}. In
  \bibinfo{booktitle}{\emph{Proceedings of the 26th International Conference on
  World Wide Web}}. \bibinfo{pages}{63--72}.
\newblock


\bibitem[Berger et~al\mbox{.}(2020)]%
        {BergerEF20}
\bibfield{author}{\bibinfo{person}{Ben Berger}, \bibinfo{person}{Alon Eden},
  {and} \bibinfo{person}{Michal Feldman}.} \bibinfo{year}{2020}\natexlab{}.
\newblock \showarticletitle{On the Power and Limits of Dynamic Pricing in
  Combinatorial Markets}. In \bibinfo{booktitle}{\emph{Web and Internet
  Economics - 16th International Conference, {WINE} 2020, Beijing, China,
  December 7-11, 2020, Proceedings}} \emph{(\bibinfo{series}{Lecture Notes in
  Computer Science}, Vol.~\bibinfo{volume}{12495})},
  \bibfield{editor}{\bibinfo{person}{Xujin Chen}, \bibinfo{person}{Nikolai
  Gravin}, \bibinfo{person}{Martin Hoefer}, {and} \bibinfo{person}{Ruta Mehta}}
  (Eds.). \bibinfo{publisher}{Springer}, \bibinfo{pages}{206--219}.
\newblock
\urldef\tempurl%
\url{https://doi.org/10.1007/978-3-030-64946-3\_15}
\showDOI{\tempurl}


\bibitem[Birge et~al\mbox{.}(2021)]%
        {birge2021optimal}
\bibfield{author}{\bibinfo{person}{John Birge}, \bibinfo{person}{Ozan
  Candogan}, \bibinfo{person}{Hongfan Chen}, {and} \bibinfo{person}{Daniela
  Saban}.} \bibinfo{year}{2021}\natexlab{}.
\newblock \showarticletitle{Optimal commissions and subscriptions in networked
  markets}.
\newblock \bibinfo{journal}{\emph{Manufacturing \& Service Operations
  Management}} \bibinfo{volume}{23}, \bibinfo{number}{3}
  (\bibinfo{year}{2021}), \bibinfo{pages}{569--588}.
\newblock


\bibitem[Bloom et~al\mbox{.}(2021)]%
        {bloom2021impact}
\bibfield{author}{\bibinfo{person}{Nicholas Bloom}, \bibinfo{person}{Robert~S
  Fletcher}, {and} \bibinfo{person}{Ethan Yeh}.}
  \bibinfo{year}{2021}\natexlab{}.
\newblock \bibinfo{booktitle}{\emph{The impact of COVID-19 on US firms}}.
\newblock \bibinfo{type}{{T}echnical {R}eport}. \bibinfo{institution}{National
  Bureau of Economic Research}.
\newblock


\bibitem[Caillaud and Jullien(2003)]%
        {caillaud2003chicken}
\bibfield{author}{\bibinfo{person}{Bernard Caillaud} {and}
  \bibinfo{person}{Bruno Jullien}.} \bibinfo{year}{2003}\natexlab{}.
\newblock \showarticletitle{Chicken \& egg: Competition among intermediation
  service providers}.
\newblock \bibinfo{journal}{\emph{RAND journal of Economics}}
  (\bibinfo{year}{2003}), \bibinfo{pages}{309--328}.
\newblock


\bibitem[Calder(2021)]%
        {FeeCapNews}
\bibfield{author}{\bibinfo{person}{Rich Calder}.}
  \bibinfo{year}{2021}\natexlab{}.
\newblock \showarticletitle{Permanent Cap on Delivery-App Fees Proposed for New
  York City}.
\newblock \bibinfo{journal}{\emph{The Wall Street Journal}}
  (\bibinfo{year}{2021}).
\newblock
\urldef\tempurl%
\url{https://www.wsj.com/articles/permanent-cap-on-delivery-app-fees-proposed-for-new-york-city-11624549165}
\showURL{%
\tempurl}


\bibitem[D'Amico-Wong et~al\mbox{.}(2024)]%
        {d2024disrupting}
\bibfield{author}{\bibinfo{person}{Luca D'Amico-Wong},
  \bibinfo{person}{Yannai~A Gonczarowski}, \bibinfo{person}{Gary Qiurui~Ma},
  {and} \bibinfo{person}{David~C Parkes}.} \bibinfo{year}{2024}\natexlab{}.
\newblock \showarticletitle{Disrupting Bipartite Trading Networks: Matching for
  Revenue Maximization}.
\newblock \bibinfo{journal}{\emph{arXiv e-prints}} (\bibinfo{year}{2024}),
  \bibinfo{pages}{arXiv--2406}.
\newblock


\bibitem[Elliott(2015)]%
        {elliott2015inefficiencies}
\bibfield{author}{\bibinfo{person}{Matthew Elliott}.}
  \bibinfo{year}{2015}\natexlab{}.
\newblock \showarticletitle{Inefficiencies in networked markets}.
\newblock \bibinfo{journal}{\emph{American Economic Journal: Microeconomics}}
  \bibinfo{volume}{7}, \bibinfo{number}{4} (\bibinfo{year}{2015}),
  \bibinfo{pages}{43--82}.
\newblock


\bibitem[Even-Dar et~al\mbox{.}(2007)]%
        {even2007network}
\bibfield{author}{\bibinfo{person}{Eyal Even-Dar}, \bibinfo{person}{Michael~J
  Kearns}, {and} \bibinfo{person}{Siddharth Suri}.}
  \bibinfo{year}{2007}\natexlab{}.
\newblock \showarticletitle{A network formation game for bipartite exchange
  economies}. In \bibinfo{booktitle}{\emph{SODA}}. \bibinfo{pages}{697--706}.
\newblock


\bibitem[Felix et~al\mbox{.}(2020)]%
        {felix2020us}
\bibfield{author}{\bibinfo{person}{Ignacio Felix}, \bibinfo{person}{Adrian
  Martin}, \bibinfo{person}{Vivek Mehta}, {and} \bibinfo{person}{Curt
  Mueller}.} \bibinfo{year}{2020}\natexlab{}.
\newblock \showarticletitle{US food supply chain: Disruptions and implications
  from COVID-19}.
\newblock \bibinfo{journal}{\emph{McKinsey \& Company, July}}
  (\bibinfo{year}{2020}).
\newblock


\bibitem[Gul and Stacchetti(1999)]%
        {gul1999walrasian}
\bibfield{author}{\bibinfo{person}{Faruk Gul} {and} \bibinfo{person}{Ennio
  Stacchetti}.} \bibinfo{year}{1999}\natexlab{}.
\newblock \showarticletitle{Walrasian equilibrium with gross substitutes}.
\newblock \bibinfo{journal}{\emph{Journal of Economic theory}}
  \bibinfo{volume}{87}, \bibinfo{number}{1} (\bibinfo{year}{1999}),
  \bibinfo{pages}{95--124}.
\newblock


\bibitem[Kakade et~al\mbox{.}(2004)]%
        {kakade2004economic}
\bibfield{author}{\bibinfo{person}{Sham~M Kakade}, \bibinfo{person}{Michael
  Kearns}, \bibinfo{person}{Luis~E Ortiz}, \bibinfo{person}{Robin Pemantle},
  {and} \bibinfo{person}{Siddharth Suri}.} \bibinfo{year}{2004}\natexlab{}.
\newblock \showarticletitle{Economic properties of social networks}.
\newblock \bibinfo{journal}{\emph{Advances in Neural Information Processing
  Systems}}  \bibinfo{volume}{17} (\bibinfo{year}{2004}).
\newblock


\bibitem[Kelso and Crawford(1982)]%
        {kelso1982job}
\bibfield{author}{\bibinfo{person}{Alexander~S Kelso} {and}
  \bibinfo{person}{Vincent~P Crawford}.} \bibinfo{year}{1982}\natexlab{}.
\newblock \showarticletitle{Job matching, coalition formation, and gross
  substitutes}.
\newblock \bibinfo{journal}{\emph{Econometrica: Journal of the Econometric
  Society}} (\bibinfo{year}{1982}), \bibinfo{pages}{1483--1504}.
\newblock


\bibitem[Kranton and Minehart(2000)]%
        {kranton2000competition}
\bibfield{author}{\bibinfo{person}{Rachel~E Kranton} {and}
  \bibinfo{person}{Deborah~F Minehart}.} \bibinfo{year}{2000}\natexlab{}.
\newblock \showarticletitle{Competition for goods in buyer-seller networks}.
\newblock \bibinfo{journal}{\emph{Review of Economic Design}}
  \bibinfo{volume}{5}, \bibinfo{number}{3} (\bibinfo{year}{2000}),
  \bibinfo{pages}{301--331}.
\newblock


\bibitem[Kranton and Minehart(2001)]%
        {kranton2001theory}
\bibfield{author}{\bibinfo{person}{Rachel~E Kranton} {and}
  \bibinfo{person}{Deborah~F Minehart}.} \bibinfo{year}{2001}\natexlab{}.
\newblock \showarticletitle{A theory of buyer-seller networks}.
\newblock \bibinfo{journal}{\emph{American economic review}}
  \bibinfo{volume}{91}, \bibinfo{number}{3} (\bibinfo{year}{2001}),
  \bibinfo{pages}{485--508}.
\newblock


\bibitem[Rabanca(2016)]%
        {maxweightsubmod}
\bibfield{author}{\bibinfo{person}{George~Octavian Rabanca}.}
  \bibinfo{year}{2016}\natexlab{}.
\newblock \bibinfo{title}{Maximum Weight Matching and submodular functions}.
\newblock
  \bibinfo{howpublished}{\url{https://cstheory.stackexchange.com/questions/34654/maximum-weight-matching-and-submodular-functions}}.
\newblock
\newblock
\shownote{Accessed: 2023-08-01}.


\bibitem[Raj et~al\mbox{.}(2020)]%
        {raj2020covid}
\bibfield{author}{\bibinfo{person}{Manav Raj}, \bibinfo{person}{Arun
  Sundararajan}, {and} \bibinfo{person}{Calum You}.}
  \bibinfo{year}{2020}\natexlab{}.
\newblock \showarticletitle{COVID-19 and digital resilience: Evidence from Uber
  Eats}.
\newblock \bibinfo{journal}{\emph{arXiv preprint arXiv:2006.07204}}
  (\bibinfo{year}{2020}).
\newblock


\bibitem[Rochet and Tirole(2003)]%
        {rochet2003platform}
\bibfield{author}{\bibinfo{person}{Jean-Charles Rochet} {and}
  \bibinfo{person}{Jean Tirole}.} \bibinfo{year}{2003}\natexlab{}.
\newblock \showarticletitle{Platform competition in two-sided markets}.
\newblock \bibinfo{journal}{\emph{Journal of the european economic
  association}} \bibinfo{volume}{1}, \bibinfo{number}{4}
  (\bibinfo{year}{2003}), \bibinfo{pages}{990--1029}.
\newblock


\bibitem[Wang et~al\mbox{.}(2023)]%
        {WangMELTZP23}
\bibfield{author}{\bibinfo{person}{Xintong Wang}, \bibinfo{person}{Gary~Qiurui
  Ma}, \bibinfo{person}{Alon Eden}, \bibinfo{person}{Clara Li},
  \bibinfo{person}{Alexander Trott}, \bibinfo{person}{Stephan Zheng}, {and}
  \bibinfo{person}{David~C. Parkes}.} \bibinfo{year}{2023}\natexlab{}.
\newblock \showarticletitle{Platform Behavior under Market Shocks: {A}
  Simulation Framework and Reinforcement-Learning Based Study}. In
  \bibinfo{booktitle}{\emph{Proceedings of the {ACM} Web Conference 2023, {WWW}
  2023, Austin, TX, USA, 30 April 2023 - 4 May 2023}},
  \bibfield{editor}{\bibinfo{person}{Ying Ding}, \bibinfo{person}{Jie Tang},
  \bibinfo{person}{Juan~F. Sequeda}, \bibinfo{person}{Lora Aroyo},
  \bibinfo{person}{Carlos Castillo}, {and} \bibinfo{person}{Geert{-}Jan
  Houben}} (Eds.). \bibinfo{publisher}{{ACM}}, \bibinfo{pages}{3592--3602}.
\newblock
\urldef\tempurl%
\url{https://doi.org/10.1145/3543507.3583523}
\showDOI{\tempurl}


\bibitem[Zhang and Conitzer(2020)]%
        {ZhangC20}
\bibfield{author}{\bibinfo{person}{Hanrui Zhang} {and} \bibinfo{person}{Vincent
  Conitzer}.} \bibinfo{year}{2020}\natexlab{}.
\newblock \showarticletitle{Learning the Valuations of a k-demand Agent}. In
  \bibinfo{booktitle}{\emph{Proceedings of the 37th International Conference on
  Machine Learning, {ICML} 2020, 13-18 July 2020, Virtual Event}}
  \emph{(\bibinfo{series}{Proceedings of Machine Learning Research},
  Vol.~\bibinfo{volume}{119})}. \bibinfo{publisher}{{PMLR}},
  \bibinfo{pages}{11066--11075}.
\newblock
\urldef\tempurl%
\url{http://proceedings.mlr.press/v119/zhang20f.html}
\showURL{%
\tempurl}


\end{thebibliography}

\appendix

\pagebreak

\section{Using Maximum Competitive Price}\label{app:prices-min}
Theorem~\ref{thm: price_lattice} shows competitive prices form a lattice. Among all competitive equilibrium prices, the maximum price $$\overline{p}_j = W(S, B, G) - W(S\setminus \{j\},B, G) $$
is the most natural when considering sellers' decision to join the platform: when a benevolent social planner charges $\alpha=0$, all sellers joining is a Platform Equilibrium. This is because for any other set of sellers $P$ joining the platform, seller $j$'s gain from joining is $p^{\on}_j-p^{\off}_j=W(S,B,G(P)\cup\{i\})-W(S,B,G(P))\geq 0$. This is not true for any other competitive prices, making it inappropriate to analyze the welfare loss due to the platform. We formalize it into the following

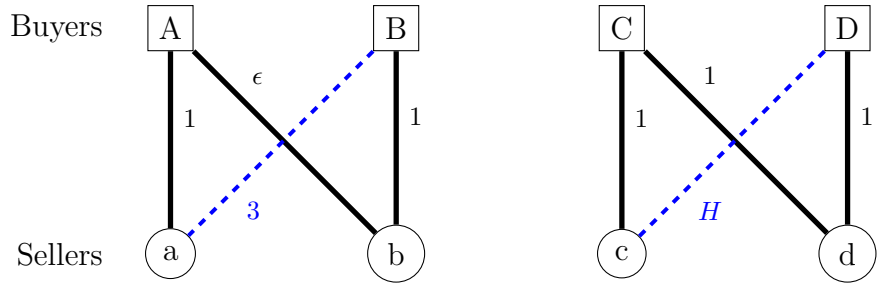
\begin{figure}[h!]
    \centering
    \begin{tikzpicture}[scale=1.5]
        % Square vertices
        \foreach \i/\label in {1/A, 3/B, 5/C, 7/D}
            \node[draw, shape=rectangle, minimum size=0.6cm] (\label) at (\i, 2) {\label};
        
        % Circle vertices
        \foreach \i/\label in {1/a, 3/b, 5/c, 7/d}
            \node[draw, shape=circle, minimum size=0.6cm] (\label) at (\i, 0) {\label};
        
        % Solid edges with weights
        \foreach \x/\y/\w/\pos in {A/a/1/midway, B/b/1/midway, C/c/1/midway, D/d/1/midway, A/b/$\epsilon$/near start, C/d/1/near start}
            \draw[line width=2pt] (\x) -- node[\pos, above right, font=\footnotesize] {\w} (\y);
        
        % Dashed blue edges with weights
        \foreach \x/\y/\w/\pos in {B/a/3/near end,  D/c/$H$/near end}
            \draw[line width=1.5pt, blue, dashed] (\x) -- node[\pos, below right, font=\footnotesize] {\w} (\y);

        % Buyers caption
        \node[left] at (0.5, 2)  {Buyers};

        \node[left] at (0.5, 0)  {Sellers};
    
    \end{tikzpicture}
    \caption{A market with infinite price of anarchy for any $\alpha\in[0,1]$.  Black solid lines are direct links, as captured by $N(i)$ for buyer $i$, and blue dotted lines indicate missing links. Buyer values are annotated adjacent to each edge, and all any value that is omitted is zero.}
    \label{fig:min_price_poa}
\end{figure} 

\begin{theorem}
    If the market clears with the minimum competitive prices, the price of anarchy of the Platform Equilibrium with pure strategies, when the platform sets a zero transaction-fee,
    %$\alpha\in[0,1]$, 
    can be arbitrarily large.
\end{theorem}
\begin{proof}
    Figure~\ref{fig:min_price_poa} gives an infinite-PoA market. When no sellers join, the allocation of goods is seller $A$ to buyer $a$, seller $B$ to buyer $b$, seller $C$ to buyer $c$, seller $D$ to buyer $d$.
    Sellers' minimum competitive prices are all zero. If seller $a$ and seller $c$ join the platform, $a$ sells to $B$, $b$  to $A$, $c$ to $D$ and $d$ to $C$. The minimum competitive prices are now $\underline{p}_a=1-\epsilon$ and  $\underline{p}_b=\underline{p}_c=\underline{p}_d=1$. So for any $\alpha\in[0,1]$ there is an equilibrium where seller $c$ does not join the platform. When $H\rightarrow\infty$, the price of anarchy approaches infinity.
\end{proof}

\section{Missing Proofs in Section~\ref{sec:existence_of_pure}}\label{app:existence_of_pure}

\subsection{None Existence of Pure Equilibrium} \label{app:non-existence}

\label{app:app_no_pure_eq}
\propNoPure*
\begin{proof}
    A three-seller-four-buyer counterexample is given in Figure~\ref{fig:no_pure_equilibrium}. At transaction fee $\alpha=\frac{1}{2}$, there exists no pure strategy equilibrium. On a high level, this is because a seller's price considers the externality she imposes on the whole market and is easily influenced by other sellers' connections. So sellers' best response dynamic can exhibit a cyclic nature: an on-platform seller can drop off after some other sellers join.

    To explicitly show there is no pure equilibrium, go through all eight pure strategy profiles: 
    \begin{enumerate}
        \item When no one joins, seller $a$ joins the platform and sell to B because $(1-\alpha)p_{a}^{on}=\frac{1}{2}(3.05-1)>1=p_{a}^{off}$.
        \item When seller $a$ is on platform, seller $b$ will have to join the platform to sell to $C$ because $B$ is taken and staying off platform brings zero price.
        \item When both seller $a,b$ are on platform, seller $c$ does not have any demand when off-platform and joins the platform in order to sell to buyer $D$.
        \item When all sellers $a,b,c$ join, seller $a$'s price decreases to $p^{on}_{a}=(3.05+1.1+0.05)-(1.1+1.15)=1.95$. Seller $a$ drops off from platform because after paying platform fees $(1-\alpha)p^{on}_{a}<1=p^{off}_{a}$ the revenue is lower than staying off platform.
        \item When $a$ drops off and $b,c$ are on platform, seller $b$'s on platform utility is $\frac{1}{2}\times1.1$ smaller than off platform utility $1$ and hence $b$ will drop.
        \item When $a,b$ are off platform, $c$ will drop as well.
        \item When $a,c$ are on platform and $b$ not, $b$ will want to join the platform.
        \item When $a,c$ are off platform and $b$ on, $b$ will drop.
    \end{enumerate} 

\end{proof}

\begin{figure}[h!] 
    \centering
    \begin{tikzpicture}[scale=1.5]
        % Square vertices
        \foreach \i/\label in {1/A, 3/B, 5/C, 7/D}
            \node[draw, shape=rectangle, minimum size=0.6cm] (\label) at (\i, 2) {\label};
        
        % Circle vertices
        \foreach \i/\label in {1/a, 3/b, 5/c}
            \node[draw, shape=circle, minimum size=0.6cm] (\label) at (\i, 0) {\label};
        
        % Solid edges with weights
        \foreach \x/\y/\w in {A/a/1, B/b/1, C/c/1}
            \draw[line width=2pt] (\x) -- node[midway, right, font=\footnotesize] {\w} (\y);
        
        % Dashed blue edges with weights
        \foreach \x/\y/\w/\pos/\loc in {B/a/3.05/midway/below right, B/c/1.15/at start/right, C/b/1.1/near start/below right, D/c/0.05/midway/below right}
            \draw[line width=1.5pt, blue, dashed] (\x) -- node[\pos, \loc, font=\footnotesize] {\w} (\y);

        % Buyers caption
        \node[left] at (0.5, 2)  {Buyers};

        \node[left] at (0.5, 0)  {Sellers};
    
    \end{tikzpicture}
    \caption*{Fig. 1. An example with no Platform Equilibrium in pure strategies at $\alpha=\frac{1}{2}$. Black solid lines are direct links, as captured by $N(i)$ for buyer $i$, and blue dotted lines indicate missing links. Buyer values are  annotated adjacent to each edge, and all any value that is omitted is zero.}
\end{figure}
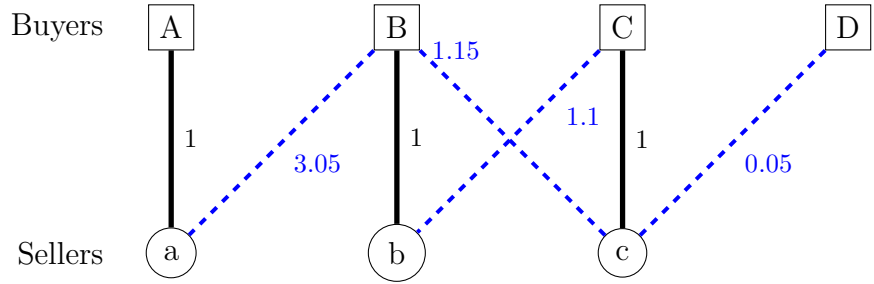

\subsection{Computing Pure Equilibrium in Homogeneous-Goods Markets} \label{app:existence}
For homogeneous-goods markets, competitive prices are directly related to a buyer's next best forgone trade opportunity, or more specifically its direct and indirect competitors. This notion is characterized by \textit{opportunity path} in \citet{kranton2000competition} as follows.
\begin{definition}[Opportunity Path \cite{kranton2000competition}]\label{def:oppo_path}
    For an allocation $\mathbf{a}$ of goods on a buyer-seller network $G$, buyer $i_1$'s %(seller $j_1$'s) 
    opportunity path linking to another buyer $i_t$ is a path $$(i_1,j_1,i_2,j_2,\ldots, j_{t-1},i_t),$$ where for every $\ell\in \{1,\ldots, t-1\}$, $$g_{i_\ell,j_\ell}=1\mbox{ and }g_{i_{\ell+1},j_\ell}=1,$$ and  $$a_{i_\ell,j_\ell}=0\mbox{ and } a_{i_{\ell+1},j_\ell}=1.$$ 
\end{definition}
\citet{kranton2000competition} show the following connection between buyers' opportunity path and sellers' maximum competitive prices, which we make use of.
\begin{theorem}[Opportunity Path Theorem \cite{kranton2000competition}]\label{thm:oppo_path}
    Consider a competitive equilibrium with maximum competitive prices $(\bar{\p},\alloc)$ where $a_{ij}=1$. Seller $j$'s price $\bar{p}_j$ is equal to the lowest valuation of any buyer linked by an opportunity path from buyer $i$. $\bar{p}_j=0$ iff buyer $i$ has an opportunity path linking to a seller who does not sell.
\end{theorem}

\citet{kranton2000competition} then analyze how adding a link between a buyer $i$ and a seller $j$ to graph $G$ affects the optimal welfare. For our purpose, we consider the case where seller $j$ has no prior links or does not transact before the new link.
\begin{lemma}\cite{kranton2000competition}\label{lem:add_one_link}
    Consider a buyer-seller network where a seller $j_0$, that either does not transact or have no prior links, is added a single buyer link to buyer $i_0$. Let $\alloc$ and $\alloc'$ be the optimal allocations before and after adding the link. In $\alloc'$ at most one previously unmatched buyer is matched. Moreover, if in $\alloc$, $a_{i_0,j_1}=1$,  $a_{i_1,j_2}=1$, \ldots, $a_{i_{t-1},j_t}=1$ (and therefore, $a_{i_\ell,j_\ell}=0$ for $\ell \in \{0,\ldots, t-1\}$), and in $\alloc'$, $a_{i_\ell,j_\ell}=1$ for $\ell \in \{0,\ldots, t-1\}$, then we have the following. Before adding the link between $j_0$ and $i_0$, $(i_t,j_t,\ldots, i_1,j_1,i_0)$ is an opportunity path. After adding the link, $(i_0, j_1,i_1,\ldots,  j_t,i_t)$ is an opportunity path. Buyer $i_t$ obtains the good in $\alloc'$ but not in $\alloc$. Furthermore, $v_{i_t}\leq v_{i_\ell}$ for every $\ell< t$.
\end{lemma}

If we introduce a set of new sellers and their edges to a market and apply Lemma~\ref{lem:add_one_link} repeatedly on each new seller, every step at most one new buyer is introduced, and thus buyers transacting in the original competitive equilibrium allocation can be preserved. This is formalized as the following corollary.
\begin{corollary}\label{cor:adding_sellers_buyers_set}
Let $(S,B,G)$ be a buyer-seller network with a corresponding optimal allocation $\alloc$, and let $B^G$ be the set of transacting buyers. Now add a set of sellers $P$ with their edges who previously know no buyers and let $G(P)$ be resulting network. %\alon{Not clear - you add the sellers with edges to all buyers?} 
For the new market $(S\cup P,B,G(P))$, there exists an optimal allocation $\alloc'$ where all buyers in $B^G$ transact.
\end{corollary}

Using this, we can relate the optimal welfare in a network to the welfare of sellers not joining the platform.

\OptimalWelfare*
\begin{proof}
    Market $(S,B,G(P))$ can be viewed as adding sellers in $P$ one-by-one with all edges to $(S\setminus P, B,G)$. When adding each new seller, Corollary~\ref{cor:adding_sellers_buyers_set} reads the previously transacting buyer remains the same, while the largest remaining unmatched buyer is matched to the new seller.
\end{proof}

\SameOnPrice*
\begin{proof}
    By Lemma~\ref{lem:optimal_welfare} and Eq.~\eqref{eq:max_price}, expand the on platform price for each seller $j\in P$ equals to $\bar{v}(m')-\bar{v}(m'-1)$, where $|P|=m'$.
\end{proof}

\vspace{0.5cm}

Algorithm~\ref{alg:PE_pure} defines a seller's on platform gain. We show it is closely related with the buyer valuation that $j$ is matched to in some markets. We will use this quantity as an intermediate step to prove Lemma~\ref{lem:algo_select} and Theorem~\ref{alg:PE_pure}. 

\begin{lemma}\label{lem:on_platform_gain}
    Consider a market $(S,B,G(P))$ where sellers $j,\hat{j}\in S, j,\hat{j}\notin P$. Let $i_P(j),i_P(\hat{j})$ be the new transacting buyer when $j$ and $\hat{j}$ are added with their off-platform edges to market $(S\setminus P \setminus \{j,\hat{j}\},B,G)$ respectively. Then $\phi_j^P > \phi_{\hat{j}}^P \Rightarrow v_{i_P(j)} < v_{i_P(\hat{j})}$, and $v_{i_P(j)} < v_{i_P(\hat{j})} \Rightarrow \phi_j^P \geq \phi_{\hat{j}}^P$
    where $\phi_j^P=(1-\alpha)p_j^{\on}(P)-p_j^{\off}(P)$ is the on-platform gain.
\end{lemma}
\begin{proof}
    Expanding the definition of prices
    \begin{eqnarray}
        \phi_j^P-\phi_{\hat{j}}^P &=& (1-\alpha)[p_j^{\on}(P)-p_{\hat{j}}^{\on}(P)]+ [p^{\off}_{\hat{j}}(P)-p^{\off}_{j}(P)]\label{eq:on_platform_gain}\\
        p_j^{\on}(P)-p_{\hat{j}}^{\on}(P) &=& W(S,B,G(P\cup \{j\}))-W(S,B,G(P\cup \{\hat{j}\}))\nonumber\\
        & + & W(S\setminus\{\hat{j}\},B,G(P)) - W(S\setminus\{j\},B,G(P)) \nonumber\\
        p^{\off}_{\hat{j}}(P)-p^{\off}_{j}(P)   &=& W(S\setminus\{j\},B,G(P)) - W(S\setminus\{\hat{j}\},B,G(P))\nonumber
    \end{eqnarray}
    View the welfare terms as adding sellers to a base market $(S\setminus P\setminus\{j,\hat{j}\}, B, G)$. For example, $W(S\setminus\{j\},B,G(P))$ is adding $\hat{j}$ to market then $P$ to platform. Let $\bar{v}_{S\setminus P\setminus \{j\}}(|P|)$ be the sum of the $|P|$ largest buyers who do not transact in $W(S\setminus P \setminus\{j\},B,G)$. (There are more than $|P|$ none-transacting buyers otherwise the algorithm would not have chosen $P$ to join.) We can thus express the welfare term as follows using Lemma~\ref{lem:optimal_welfare}
    \begin{eqnarray*}
        W(S\setminus\{j\},B,G(P)) &=& W(S\setminus P\setminus \{j,\hat{j}\},B,G)+v_{i_P(\hat{j})}+\bar{v}_{S\setminus P\setminus \{j\}}(|P|)\\
        W(S\setminus\{\hat{j}\},B,G(P)) &=& W(S\setminus P\setminus \{j,\hat{j}\},B,G)+v_{i_P(j)}+\bar{v}_{S\setminus P\setminus \{\hat{j}\}}(|P|)\\
        W(S,B,G(P\cup\{j\})) &=& W(S\setminus P\setminus \{j,\hat{j}\},B,G)+v_{i_P(\hat{j})}+\bar{v}_{S\setminus P\setminus \{j\}}(|P|+1)\\
        W(S,B,G(P\cup\{\hat{j}\})) &=& W(S\setminus P\setminus \{j,\hat{j}\},B,G)+v_{i_P(j)}+\bar{v}_{S\setminus P\setminus \{\hat{j}\}}(|P|+1)
    \end{eqnarray*}
    Let $\bar{v}_{S\setminus P\setminus\{j\}}^{|P|+1}$ be the $|P|+1$-th largest buyer who does not transact in $W(S\setminus P\setminus \{j\},B,G)$, and similarly for $\bar{v}_{S\setminus P\setminus\{\hat{j}\}}^{|P|+1}$. 
    Bring these back to Equation~\eqref{eq:on_platform_gain}
    \begin{eqnarray*}
        \phi_j^P-\phi_{\hat{j}}^P = [v_{i_P(\hat{j})}+\bar{v}_{S\setminus P\setminus \{j\}}(|P|)+(1-\alpha)\bar{v}_{S\setminus P\setminus\{j\}}^{|P|+1}]-
        [v_{i_P(j)}+\bar{v}_{S\setminus P\setminus \{\hat{j}\}}(|P|)+(1-\alpha)\bar{v}_{S\setminus P\setminus\{\hat{j}\}}^{|P|+1}]
    \end{eqnarray*}
    Starting from market $W(S\setminus P\setminus \{j,\hat{j}\},B,G)$,
    the first bracket is the value of buyer that $\hat{j}$ introduces, plus the $|P|$ largest remaining buyers, plus $(1-\alpha)$ of the $|P|+1$-th remaining largest buyer. The second bracket is similar, after $j$ introduces a buyer. To guarantee $\phi_j^P>\phi_{\hat{j}}^P$, a necessary condition is $v_{i_P(\hat{j})}>v_{i_P(j)}$. Note both inequalities are strict: if $\phi_j^P=\phi_{\hat{j}}^P$, it can be that $v_{i_P(\hat{j})}\leq v_{i_P(j)}$. On the other direction, if $v_{i_P(\hat{j})}>v_{i_P(j)}$, we can only guarantee $\phi_j^P\geq \phi_{\hat{j}}^P$.
\end{proof}

\vspace{0.1cm}

\AlgoSelect*
\begin{proof}
    Lemma~\ref{lem:on_platform_gain} shows that the necessary condition for $\phi_{\hat{j}}^P>\phi_j^P$ is $v_{i_P(j)}>v_{i_P({\hat{j}})}$. Thus for all sellers $ {\hat{j}}\in \Phi_{\max}^P$ and all sellers  $j\in S\setminus P, j\notin \Phi_{\max}^P$
    \begin{eqnarray*}
        W(S\setminus P\setminus \{{\hat{j}},j\},B,G) + v_{i_P({\hat{j}})} & < & W(S\setminus P\setminus \{{\hat{j}},j\},B,G) + v_{i_P(j)}\\
        W(S\setminus P\setminus \{j\},B,G) &< & W(S\setminus P\setminus \{{\hat{j}}\},B,G).
    \end{eqnarray*}
    Let $\bar{v}_{S\setminus P\setminus\{j\}}(|P|)$ be the sum of the $|P|$ largest buyers who do not transact in $W(S\setminus P\setminus\{j\},B,G)$, and $\bar{v}_{S\setminus P\setminus\{{\hat{j}}\}}(|P|)$ be the sum of the $|P|$ largest buyers who do not transact in $W(S\setminus P\setminus\{\hat{j}\},B,G)$.
    \begin{eqnarray*}
        W(S\setminus P\setminus \{j\},B,G) + \bar{v}_{S\setminus P\setminus\{j\}}(|P|) &\leq & W(S\setminus P\setminus \{{\hat{j}}\},B,G) + \bar{v}_{S\setminus P\setminus\{{\hat{j}}\}}(|P|) \\
        W(S\setminus \{j\},B,G(P)) &\leq & W(S\setminus \{{\hat{j}}\},B,G(P)).
    \end{eqnarray*}
    By expanding the definition of prices, we have
    \begin{eqnarray*}
        p_{\hat{j}}^{\off}(P) &=& W(S,B,G(P))-W(S\setminus\{{\hat{j}}\},B,G(P))\\
        &\leq& W(S,B,G(P))-W(S\setminus\{j\},B,G(P))=p_{j}^{\off}(P)
    \end{eqnarray*}
    For sellers ${\hat{j}},j$ within set $\Phi_{\max}^P$, $\phi_{\hat{j}}^P=\phi_{j}^P$. The algorithm in Step~\ref{alg:step} breaks ties towards ${\hat{j}}$ with the lowest off platform price.
    
    For the other direction, one can verify if $p_{\hat{j}}^{\off}(P)<p_{j}^{\off}(P)$, then $v_{i_P}(j)>v_{i_P}({\hat{j}})$. Lemma~\ref{lem:on_platform_gain} then guarantees $\phi_{\hat{j}}^P\geq \phi_{j}^P$. So ${\hat{j}}$ is in $\Phi_{\max}^P$ at Step~\ref{alg:step} of the Algorithm and selected because it has lower off platform price.
\end{proof}

\vspace{0.1in}

To prove theorem~\ref{thm:PE_pure}, we will need one more conditions. \citet{maxweightsubmod} proves max weight matching on bipartite graphs satisfies the submodularity condition: with more vertices on one side, adding one new edge brings marginally lower increase in maximum weight. We rephrase it here and apply it to the buyer-seller markets.
\begin{lemma}[Max weight matching submodularity \cite{maxweightsubmod}]\label{lem:submodularity}
    For a set of sellers $S'\subseteq S$, a single seller $j\notin S$ and buyers $B$ connection $G$   
    \begin{eqnarray}\label{eq:submodularity}
        W(S\cup\{j\},B,G)-W(S,B,G) \leq W(S'\cup\{j\},B,G)-W(S',B,G)
    \end{eqnarray}
\end{lemma}
\begin{proof}
    we prove the above through an equivalent definition of submodularity. For set of sellers $S_1,S_2\subseteq S$, denote $W(S_1):=W(S_1,B,G),W(S_2):=W(S_2,B,G)$. Prove $$W(S_1)+W(S_2)\geq W(S_1\cup S_2)+W(S_1\cap S_2)$$ 
    Using $M_{\cup},M_{\cap}$ to denote the edges in max matching with sellers $S_1\cup S_2, S_1\cap S_2$, it suffices to prove that $M_{\cup}$ and $M_{\cap}$ can be partitioned into two disjoint matchings $M_1, M_2$ for $W(S_1)$ and $W(S_2)$ respectively. Then by max weight matching $W(S_1)$ is weakly greater than the weight of matching $M_1$, $W(S_2)$ greater than that of $M_2$.
    
    Let $C$ denote the collection of edges in $M_\cap$ and $M_\cup$. Edges in $C$ form alternating paths and cycles. A path alternate between edges in $M_\cap$ and $M_\cup$, while cycles are restricted to vertices within $S_1\cap S_2$: for vertices in $S_1\setminus S_2$ and $S_2\setminus S_1$, there is at most one edge in $C$ connecting them because they are not in $M_\cap$.

    Let $L_1$ be the set of paths with at least one vertex in $S_1\setminus S_2$, $L_2$ be the set of paths with at least one vertex in $S_2\setminus S_1$. Then $L_1\cap L_2 = \emptyset$. Suppose a path $l$ contains vertex $j_1\in S_1\setminus S_2, j_2\in S_2\setminus S_1$, then $j_1,j_2$ must be two end points of $l$ because they have one edge. For the same reason, there are no other vertex of $S_1\setminus S_2, S_2\setminus S_1$ in $l$. Since $j_1, j_2$ are both in $S$, there are even number of edges in this path. This cannot be because $l$ alternates between edges in $M_\cap$ and $M_\cup$: if $l_1$ is connected through edges in $M_\cup$, by even number of edges $l_2$ must be connected through edges in $M_\cap$. Vice versa. Either $j_1$ or $j_2$ must be in $S_1\cap S_2$, forming a contradiction.

    Now we are ready to partition $M_\cap$ and $M_\cup$. Let 
    \begin{eqnarray}
        M_1  &= (L_1\cap M_\cup) \cup (L_2\cap M_\cap) \cup (C\setminus \{L_1\cup L_2\}\cap M_\cap)\\
        M_2  &= (L_2\cap M_\cup) \cup (L_1\cap M_\cap) \cup (C\setminus \{L_1\cup L_2\}\cap M_\cup)
    \end{eqnarray}
    We verify $M_1\cup M_2 = M_\cup \cup M_\cap$, $M_1\cap M_2 = M_\cup \cap M_\cap$, confirming $M_1$ and $M_2$ are indeed two partitions of $M_\cap, M_\cup$. Now prove $M_1$ is a matching for $W(S_1)$. First $M_1$ uses no vertex from $S_2\setminus S_1$: $L_1$, $M_\cap$ and $C\setminus\{L_1,L_2\}$ contains no vertex from $S_2\setminus S_1$. Second the three terms in $M_1$ being independent, each vertex in $S_1$ is matched at most once in $M_1$. Otherwise, a vertex would be matched twice in either $M_\cup$ or $M_\cap$, contradiction. Similarly, prove $M_2$ is a matching for $W(S_2)$, concluding the proof.
\end{proof}

Now we are ready to prove Theorem~\ref{thm:PE_pure}.
\PEpure*
\begin{proof}

The algorithm always terminates because there are finite sellers. To prove it terminates at a pure strategy platform equilibrium, it suffices to show seller $j\in P$ already on the platform don't have incentive to drop off when a new seller $\hat{j}$ is added to $P$. In particular, we want to prove $\forall j\in P$, its benefit of staying on platform %$\phi_j^{P\cup\{\hat{j}\}}$ 
    is larger than $\hat{j}$'s benefit of joining.
    \begin{eqnarray}
        \forall j\in P, \phi_j^{P\cup\{\hat{j}\}}\geq \phi_{\hat{j}}^{P}\geq 0 \text{ for } \hat{j}\in \argmax_{j\in S\setminus P}\{\phi_j^P\};\label{eq:pure_eq_invariant}
    \end{eqnarray}
    Then sellers in $P$ do not drop off from platform when the algorithm terminates. 
    Since $\hat{j}$ is the next seller joining, $\phi_{\hat{j}}^{P}\geq 0$. Lemma~\ref{lem:on_platform_gain} does not apply because $j\in P$. So we expand Eq.~\eqref{eq:pure_eq_invariant} to prove the following is non-negative
    \begin{eqnarray}
        \phi_j^{P\cup\{\hat{j}\}} - \phi_{\hat{j}}^{P} = (1-\alpha)[p_j^{\on}(P\cup\{\hat{j}\})-p_{\hat{j}}^{\on}(P)]+[p_{\hat{j}}^{\off}(P)-p_{j}^{\off}(P\cup\{\hat{j}\})]\label{eq:pure_eq_invariant_2}
    \end{eqnarray}
    The first term $p_j^{\on}(P\cup\{\hat{j}\})-p_{\hat{j}}^{\on}(P)$ in Eq.~\eqref{eq:pure_eq_invariant_2} is zero because according to Corollary~\ref{lem:same_on_price}, all on-platform sellers in a same market have the same price. Now prove $p_{\hat{j}}^{\off}(P) \geq p_{j}^{\off}(P\cup\{\hat{j}\})$. 
    Intuitively this is true because $j$ is selected by the algorithm before $\hat{j}$ and had a weakly lower off-platform price when it was added. For the rest of the proof, we will prove this intuition.
    
    Let $P_{-j}$ be $P\setminus \{j\}$. Expand the two prices $p_{\hat{j}}^{\off}(P)=W(S,B,G(P))-W(S\setminus \{\hat{j}\},B,G(P))$ and $p_{j}^{\off}(P\cup\{\hat{j}\}) = W(S,B,G(P_{-j}\cup \{\hat{j}\}))-W(S\setminus\{j\},B,G(P_{-j}\cup \{\hat{j}\}))$. The two minus terms are equal $$W(S\setminus \{\hat{j}\},B,G(P))=W(S\setminus\{j\},B,G(P_{-j}\cup \{\hat{j}\}))$$ because the two terms have exactly the same off platform sellers and the same number of on-platform sellers. By Lemma~\ref{lem:optimal_welfare}, they have the same welfare. Rearranging, prove the following is non-negative.
    \begin{eqnarray}
        p_{\hat{j}}^{\off}(P) - p_{j}^{\off}(P\cup\{\hat{j}\}) = W(S,B,G(P))-W(S,B,G(P_{-j}\cup \{\hat{j}\}))\label{eq:pure_eq_target_price}
    \end{eqnarray}

    Just as in the proof for Lemma~\ref{lem:on_platform_gain}, view $W(S,B,G(P))$ as adding seller $\hat{j}$ to the market, then $P$ to platform in a base market $(S\setminus P\setminus\{\hat{j}\},B,G)$. Let $i_P(\hat{j})$ be the new buyer introduced by $\hat{j}$, $\bar{v}_{S\setminus P}(|P|)$ be sum of the $|P|$ largest buyers who do not transact in $W(S\setminus P,B,G)$. Similarly, let $i_P(j)$ be the new buyer introduced by adding $j$ to the base market, $\bar{v}_{S\setminus P_{-j}\setminus\{\hat{j}\}}(|P|)$ be sum of the $|P|$ largest buyers who do not transact in $W(S\setminus P_{-j}\setminus\{\hat{j}\},B,G)$. There are more than $|P|$ buyers not transacting in $W(S\setminus P,B,G)$ or $W(S\setminus P_{-j}\setminus\{\hat{j}\},B,G)$, otherwise all buyers transact in $(S,B,G(P))$, contradicting with the algorithm picking $\hat{j}$ next. Assume for contradiction Equation~\eqref{eq:pure_eq_invariant_2} and \ref{eq:pure_eq_target_price} is negative, expanding 
    \begin{eqnarray*}
        \phi_j^{P\cup\{\hat{j}\}} - \phi_{\hat{j}}^{P} = v_{i_P(\hat{j})}+\bar{v}_{S\setminus P}(|P|)-v_{i_P(j)}-\bar{v}_{S\setminus P_{-j}\setminus\{\hat{j}\}}(|P|) < 0
    \end{eqnarray*}
    A necessary condition for it being negative is 
    \begin{eqnarray}\label{eq:pure_eq_target_price_decompose}
        v_{i_P(\hat{j})} & < v_{i_P(j)} 
    \end{eqnarray}
    Now consider right before Algorithm~\ref{alg:PE_pure} adds $j$ to the platform. Denote the sellers joining the platform then by $P^{j}\subset P$. The Algorithm chooses $j$ at Step~\ref{alg:step}. Let $i_{P^j}(j), i_{P^j}(\hat{j})$ be the new transacting buyer when $j$ and $\hat{j}$ is added with their off-platform edges to market $(S\setminus P^j\setminus \{j,\hat{j}\},B,G)$ respectively.
    
    If $\phi_j^{P^j}=\phi_{\hat{j}}^{P^j}$ and $p^{\off}_j(P^j)=p^{\off}_{\hat{j}}(P^j)$. There are two cases.
    
    \textbf{Case 1:} $j$ and $\hat{j}$'s off platform prices are determined by two different buyers of the same valuation on two different opportunity paths. Then $j$ joining the platform won't change $\hat{j}$'s off platform price. So $\hat{j}$ is the next to join. Similarly, $\hat{j}$ joining won't change $j$'s off platform price and equation~\eqref{eq:pure_eq_target_price} equals to zero.
    
    \textbf{Case 2:} $j$ and $\hat{j}$'s off platform prices equal to the the valuation of the smallest buyer on the same opportunity path. This valuation equals to $p_j^{\off}(P^j)$. By Lemma~\ref{lem:optimal_welfare},
    $W(S,B,G(P^j))=W(S\setminus P^j,B^G,G)+\bar{v}(|P^j|)$. Sellers in $P^j$ are matched to largest buyers in $B\setminus B^G$. If the smallest buyer in $j$ and $\hat{j}$'s opportunity path is matched to an on platform seller in $P^j$, the next seller joining the platform will have a price weakly lower than $p_j^{\off}(P^j)$. This contradicts with $j$ joining the platform next. If an on platform seller is on $j$ and $\hat{j}$'s opportunity path but not matched to the smallest buyer on the path, one can remove this on platform seller, shift  transactions along the path and strictly increase $W(S\setminus P^j, B^G,G)$. This contradicts with off-platform sellers being matched optimally. So no seller in $P^j$ is on $j$ and $\hat{j}$'s opportunity path. 
    %
    % $l=(j_0,i_0,j_1,...,j_t,(i_t))$.One of the two sellers $j,\hat{j}$ is at the start of the path ($j_0=\hat{j}$ or $j_0=j$), and the other one can be $j_1, j_2,...,j_t$.
    Now, if the opportunity path ends at a seller who does not transact, then $v_{i_{P^j}(\hat{j})} = v_{i_{P^j}(j)}$ equals to the lowest buyer valuation on the path. If the opportunity path ends on the buyer with the lowest valuation, then $v_{i_{P^j}(\hat{j})} = v_{i_{P^j}(j)}$ equals to the second smallest buyer valuation on the path. This is discussed together with the next case in inequality~\eqref{eq:pure_eq_cond_price_decompose}.

    Otherwise $\phi_j^{P^j}>\phi_{\hat{j}}^{P^j}$ or $p^{\off}_j(P^j)<p^{\off}_{\hat{j}}(P^j)$. Lemma~\ref{lem:on_platform_gain} reads \begin{eqnarray}\label{eq:pure_eq_cond_price_decompose}
        v_{i_{P^j}(j)} \leq v_{i_{P^j}(\hat{j})} 
    \end{eqnarray}
    Let $P^{\text{mid}}=P_{-j}\setminus P^j$. The two base markets $S\setminus P\setminus \{\hat{j}\}$ and $S\setminus P^j\setminus\{j,\hat{j}\}$ differ by $P^{\text{mid}}$. When $P^{\text{mid}}=\emptyset$, the two base markets are the same, $v_{i_P(\hat{j})}=v_{i_{P^j}(\hat{j})}, v_{i_P(j)}=v_{i_{P^j}(j)}$. Inequality~\eqref{eq:pure_eq_target_price_decompose} directly contradicts with Inequality~\eqref{eq:pure_eq_cond_price_decompose}. So the difference in on-platform gain in equation~\eqref{eq:pure_eq_invariant_2} cannot be negative.
    
    We now analyze when $P^{\text{mid}}\neq \emptyset$. The intuition is according to inequality~\eqref{eq:pure_eq_target_price_decompose}, $i_P(j)$ as a buyer of high value is matched to $j$ in $(S\setminus P\setminus\{\hat{j}\},B,G)$. However, it is matched to another seller $j'\in P^{\text{mid}}$ in market $(S\setminus P^j\setminus\{j,\hat{j}\},B,G)$ because of inequality~\eqref{eq:pure_eq_invariant_2}. Then when adding $j'$ on platform, the algorithm should have picked $\hat{j}$ instead because $\hat{j}$ introduces a low valuation buyer. We now formalize this intuition.
    
    Comparing the two base markets, $S\setminus P\setminus \{\hat{j}\} \subset S\setminus P^j\setminus\{j,\hat{j}\}$, applying the submodularity condition in Lemma~\ref{lem:submodularity}
    \begin{eqnarray}
        v_{i_P(\hat{j})} & \geq v_{i_{P^j}(\hat{j})} \label{eq:apply_submodularity_to_target}\\
        v_{i_P(j)} & > v_{i_{P^j}(j)}\label{eq:apply_strict_submodularity_to_condition}
    \end{eqnarray}
    Inequality~\eqref{eq:apply_strict_submodularity_to_condition} is strict because otherwise combining inequality~\eqref{eq:pure_eq_target_price_decompose} and \eqref{eq:pure_eq_cond_price_decompose}
    we have $v_{i_{P^j}(\hat{j})} \geq v_{i_{P^j}(j)} = v_{i_P(j)} > v_{i_P(\hat{j})}$, which contradicts with \eqref{eq:apply_submodularity_to_target}. By definition, $j$ introduces $i_{P}(j)$ in the base market $(S\setminus P\setminus\{\hat{j}\},B,G)$. Thus by Corollary~\ref{cor:adding_sellers_buyers_set}, $i_P(j)$ still transacts in the second base market containing more sellers $(S\setminus P^j\setminus\{\hat{j}\},B,G)$. But the strict inequality ~\eqref{eq:apply_strict_submodularity_to_condition} says in this market, $j$ does not introduce $i_p(j)$ anymore. So $i_P(j)$ still transacts in the market $(S\setminus P^j\setminus\{j\},B,G)$. Note by definition $i_P(j)$ does not transact in market $(S\setminus P,B,G)$.
    
    Now examine when the algorithm adds seller $j$ to platform, and before it adds seller $\hat{j}$ to platform. The former market is $(S,B,G(P^j\cup \{j\}))$, where sellers in $S\setminus P^j\setminus\{j\}$ is off-platform. The latter market is $(S,B,G(P))$, where $S\setminus P$ is off-platform. By Corollary~\ref{cor:adding_sellers_buyers_set} and our discussion in the previous paragraph, $i_P(j)$ transacts with off-platform sellers in the former market but not in the latter. Then there $\exists j'\in P\setminus P^j$ such that it joining the platform results in $i_P(j)$ not transacting with off-platform sellers. Equivalently, adding $j'$ to off-platform sellers will introduce $i_P(j)$.
    
    Consider when the algorithm adds $j'$ to platform. By Lemma~\ref{lem:add_one_link}, $i_P(j)$ is the lowest-value buyer on $j'$'s opportunity path so $p^{\off}(j')=v_{i_P(j)}$. 
    Again because of submodularity, seller $\hat{j}$ introduces a buyer with valuation smaller than $v_{i_P(\hat{j})}$ so its off-platform price is $p^{\off}_{\hat{j}}\leq v_{i_P(\hat{j})}<p^{\off}_{j'}$. By Lemma~\ref{lem:algo_select} the algorithm should not have added seller $j'$, which forms a contradiction.
\end{proof}

\vspace{0.5cm}
The algorithm can terminate any time when $\forall j\notin P, \phi_j^P\leq 0$. This gives flexibility to find multiple pure equilibrium, particularly of interest is when $\alpha=1$. 
\PEpureAlpha*
\begin{proof}
    At $\alpha=1$ sellers in $S^G$ weakly prefer not to join. After adding a first $j_1\in \bar{S}^G$ seller to $P$, other sellers in $\bar{S}^G\setminus \{j_1\}$ still does not transact because Lemma~\ref{lem:add_one_link} reads adding a new link involves no new sellers. Thus, every seller in $\bar{S}^G$ is indifferent to join at $\alpha=1$.
\end{proof}

This following lemma allows the platform to continuously decrease $\alpha$ with algorithm~\ref{alg:PE_pure}, and gradually increase the set of on-platform sellers in equilibrium to calculate revenue.
\PEpureContinuous*
\begin{proof}
At $\alpha_2$, Step~\ref{alg:step} can first add all sellers in $P_1$ before any other sellers: a lower $\alpha$ does not change the seller with largest on-platform gain. 
\begin{eqnarray*}
    (1-\alpha_1)p_j^\on(P)-p_j^\off(P)\geq (1-\alpha_1)p_{j'}^\on(P)-p_{j'}^\off(P)\\
    \Rightarrow (1-\alpha_2)p_j^\on(P)-p_j^\off(P)\geq (1-\alpha_2)p_{j'}^\on(P)-p_{j'}^\off(P)
\end{eqnarray*}
\end{proof}

\EnlargeEq*
\begin{proof}
    Lemma~\ref{lem:pe_pure_alpha1} proves a pure platform equilibrium always exists for $m_p=1,...,\bar{S}^G$ at $\alpha=1$. At $P=\bar{S}^G$, for all sellers $j\notin \bar{S}^G, \phi_j^{\bar{S}^G}=-p_j^{\off}(\bar{S}^G)\leq 0$. As $\alpha$ decreases continuously from $1$ to $0$, $\phi_j^P$ continuously increases for all $j$ as long as the equilibrium set $P$ stays the same. When a first set of sellers $\Phi_{\max}^P$ have zero on-platform gain, Algorithm~\ref{alg:PE_pure} adds a seller $j\in \Phi_{\max}^P$ to platform. Then it suffices to show after $j$ is added to $P$, $\forall \hat{j}\notin P$, either
    \begin{eqnarray}
    \phi_{\hat{j}}(P\cup \{j\}) & \leq & \phi_{\hat{j}}(P) \nonumber\\
    \Leftrightarrow	 (1-\alpha)[p_{\hat{j}}^{\on}(P\cup\{j\})-p_{\hat{j}}^{\on}(P)] &\leq 
    &  p_{\hat{j}}^{\off}(P\cup\{j\})-p_{\hat{j}}^{\off}(P)\label{eq:equilibrium_set_enlarges_1}
    \end{eqnarray}
    or $\phi_{\hat{j}}(P\cup\{j\})<0$.
    Then for $\hat{j}\notin P\cup\{j\}$, on platform gain is non-positive. So $P\cup \{j\}$ is indeed a pure equilibrium of size one larger than $P$. And the algorithm can further add sellers whose on platform gain is zero, and after that continue to decrease $\alpha$.

    We first look at the left hand side of equation~\eqref{eq:equilibrium_set_enlarges_1}, the change in on platform prices. By Lemma~\ref{lem:same_on_price}, $p_{\hat{j}}^{\on}(P\cup\{j\})$ equals to the value of the $(|P|+2)$-st largest buyer who does not transact in market $(S\setminus P\setminus \{\hat{j},j\},B,G)$; $p_{\hat{j}}^{\on}(P)$ equals to the value of the $(|P|+1)$-st largest buyer who does not transact in market $(S\setminus P\setminus \{\hat{j}\},B,G)$. There is possibly one more none-transacting buyer in the former market than the latter, so $p_{\hat{j}}^{\on}(P\cup\{j\}) \leq p_{\hat{j}}^{\on}(P)$.

    Now we examine the right hand side of equation~\eqref{eq:equilibrium_set_enlarges_1}, the change in $\hat{j}$'s off platform price after $j$ joins.
    If $p_{\hat{j}}^{\off}(P\cup\{j\}) \geq p_{\hat{j}}^{\off}(P)$, naturally equation~\eqref{eq:equilibrium_set_enlarges_1} holds. We first find and rule out the sellers with weakly larger off platform prices. We then prove sellers $\hat{j}$ with smaller off platform price will be matched to the same buyer even joining the platform, thus $\phi_{\hat{j}}(P\cup\{j\})<0$.
    %We will analyze it through changes in the opportunity path after $j$ joins platform.
    
    View seller $j$ joining the platform as first leaving the market $(S,B,G(P))$, and then transacting with the largest remaining buyer. By Lemma~\ref{lem:add_one_link} at most one buyer $i^{G(P)}(j)$ is unmatched because of $j$ leaving, and it is connected through opportunity path $$l_{G(P)}=(j_0=j,i_0,j_1,i_1,...,j_t,i_t=i^{G(P)}(j))$$ in $(S,B,G(P))$ where $a_{j_q,i_q}=1 \text{ for } q\in\{0,...,t\}$.  Furthermore, $i^{G(P)}$ is the smallest buyer on the opportunity path $p_j^{\off}(P)=v_{i^{G(P)}(j)}$.
    When $j$ joins the platform, it is connected to all buyers. Again by Corollary~\ref{lem:optimal_welfare}, it sells to the largest none-transacting buyer $i^{G(P\cup\{j\})}(j)$. By Lemma~\ref{lem:same_on_price}, $p_j^{\on}(P)=v_{i^{G(P\cup\{j\})}(j)}$. This in turn means $i^{G(P\cup\{j\})}(j)$ is the smallest value buyer in any opportunity path that extends from $i^{G(P\cup\{j\})}(j)$. We also have $v_{i^{G(P\cup\{j\})}(j)}>v_{i^{G(P)}(j)}$ otherwise seller $j$ would not have joined the platform.
    % changes in oppo path analyze.
    When $j$ joins the platform, opportunity path $l_{G(P)}$ changes and reverses into $$l_{G(P\cup\{j\})}=(j_t,i_{t-1},...,i_1,j_1,i_0,j_0=j,i^{G(P\cup\{j\})}(j))$$
    where $a_{j_0,i^{G(P\cup\{j\})}(j))}=1$ and $a_{j_q,i_{q-1}}=1$ for $q\in\{t,t-1,...,1\}$.
    Since price equals the lowest buyer valuation in opportunity path, sellers $j_1,j_2,...,j_t$'s off platform price in $(S,B,G(P))$ is $v_{i^{G(P)}(j)}$ and that in $(S,B,G(P\cup\{j\}))$ is weakly larger because $v_{i^{G(P\cup\{j\})}(j)}>v_{i^{G(P)}(j)}$, and $i^{G(P)}(j)$ is the smallest buyer in $l_{G(P)}$. So we have proved for $\hat{j}\in \{j_1,...,j_t\}$, off platform price weakly increases.

    We now show if off platform prices decreases for $\hat{j}\notin \{j_1,...,j_t\}$, they will be matched to the same buyer even joining the platform. Seller $\hat{j}$ transacts in $(S,B,G(P\cup\{j\}))$ with the same buyer as $(S,B,G(P))$ off-platform. Denote the buyer as $i(\hat{j})$. There are two possible changes to opportunity paths that begins at $i(\hat{j})$. First, if $i(\hat{j})$ is connected to $l_{G(P)}$ by an opportunity path, then this path will instead connect to $l_{G(P\cup\{j\})}$. The minimum buyer valuation on this opportunity path increases. Second, now that $j$ joins the platform, $i(\hat{j})$ knows $j$ and have an opportunity path to $i^{G(P\cup\{j\})}(j)$. If $i(\hat{j})$'s off platform price decreases, the smallest buyer connected through opportunity path must be $i^{G(P\cup\{j\})}(j)$, meaning  $$v_{i(\hat{j})}>p_{\hat{j}}^{\off}(P\cup\{j\})= v_{i^{G(P\cup\{j\})}(j)}$$ By Lemma~\ref{lem:same_on_price}, if $\hat{j}$ joins the platform, it is matched to the largest none-transacting buyer, which is of weakly smaller value than $i^{G(P\cup\{j\})}(j)$, and smaller than $i(\hat{j})$. So optimal matching requires $\hat{j}$ be matched to $i(\hat{j})$ even on platform. Then $\phi_{\hat{j}}(P\cup\{j\})<0$ and $\hat{j}$ doesn't want to join the platform after $j$ joins.
\end{proof}

\section{$\min(m,n)$ Lower Bound for the Price of Anarchy in General Markets} \label{app:poa_mn_general}

In this section, we show that without regulation, in general unit-demand markets, the price of anarchy can be as bad as $\min\{m,n\}$. This is in stark contrast to the positive results we have in Section~\ref{sec:poa-regulated} for the case the platform is regulated in the fees it can set.

We show this through an example market in Figure~\ref{fig:poa_mn_general}. There are two main reasons for large PoA in this general valuation market. First, the unique market structure allows for only two pure equilibria. The platform either posts a low transaction fee where all sellers join, or a high fee where only one joins. Second, even though the equilibrium where all join has high social welfare, each seller has low price. Mixed equilibrium requires too low a transaction fee. The platform thus selects the pure equilibrium where only one seller joins.   

\begin{figure} 
    \centering
    \begin{tikzpicture}[scale=1.5]
        % Square vertices
        \foreach \i/\label in {1/$b_1$, 3/$b_2$, 7/$b_{n-1}$, 9/$b_n$}
            \node[draw, shape=rectangle, minimum size=0.6cm] (\label) at (\i, 2) {\label};
        
        % Circle vertices
        \foreach \i/\label in {1/$s_1$, 3/$s_2$, 7/$s_{n-1}$, 9/$s_n$}
            \node[draw, shape=circle, minimum size=0.6cm] (\label) at (\i, 0) {\label};
            
        \foreach \i\j in {$b_1$/$s_1$,$b_2$/$s_2$,$b_{n-1}$/$s_{n-1}$}
            \draw[line width=1.2pt] (\i) -- (\j);

        % Ellipsis
        \node at (5, 2) {$\ldots$};
        \node at (5, 0) {$\ldots$};        

        % Dashed blue edges with weights
        \foreach \x/\y/\w/\pos/\loc in {$b_1$/$s_n$/1/at start/below right,  $b_2$/$s_n$/$x$/at start/below right, $b_{n-1}$/$s_n$/$x$/near start/below left, $b_n$/$s_n$/$x$/near start/below right}
            \draw[line width=1.1pt, blue, dashed] (\x) -- node[\pos, \loc, font=\footnotesize] {\w} (\y);

        \draw[line width=1.1pt, blue, dashed] (8.8,0.2) -- node[near end,above,font=\footnotesize] {$x$} (5.2,1.8);

        \foreach \x/\y/\w/\pos/\loc in {$s_1$/$b_2$/$\frac{n}{n-1}x$/near start/below right, $s_{n-1}$/$b_n$/$\frac{n}{n-1}x$/at start/below right}
            \draw[line width=1.1pt, blue, dashed] (\x) -- node[\pos, \loc, font=\footnotesize] {\w} (\y);

        \draw[line width=1.1pt, blue, dashed] (3.2,0.2) -- node[near start,below right,font=\footnotesize] {$\frac{n}{n-1}x$} (4.8,1.8);

        % captions
        \node[left] at (0.5, 2)  {Buyers};
        \node[left] at (0.5, 0)  {Sellers};
    \end{tikzpicture}
    \caption{An $n$-buyer-$n$-seller general valuation market where $PoA=n$. Black solid lines are direct links, as captured by $N(i)$ for buyer $i$, and blue dotted lines indicate missing links. Buyer values are annotated adjacent to each edge. $x$ is a large constant. For the first seller $v_{11}=\frac{nx}{(n-1)(x+n)}+\frac{n}{x+n}, v_{21}=\frac{nx}{n-1}$. For sellers $j=2,...,n-1$, $v_{j+1,j}=\frac{n}{n-1}x, v_{jj}=\frac{n^2x}{(n-1)(x+n)}+(i-1)\epsilon$. For the last seller $v_{1n}=1,v_{jn}=x$ for $j=2,3,...,n$   All other values are zero.}
    \label{fig:poa_mn_general}
\end{figure}
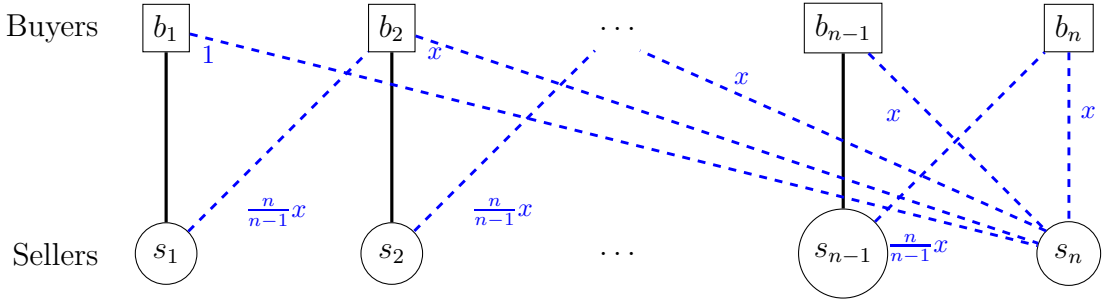 

\begin{restatable}{proposition}{propPoAMNGeneral} \label{prop:poa_mn_general}
    There exists a $n$-buyer-$m$-seller  market for which the price of anarchy is $\min(\{n,m\})$.
\end{restatable}
\begin{proof}
Figure~\ref{fig:poa_mn_general} gives a $n$-seller-$n$-buyer market where PoA approaches $n$. Black solid lines are direct links, as captured by $N(i)$ for buyer $i$, and blue dotted lines indicate missing links. Seller $n$ has no off-platform links and always join the platform, and seller $i=1,2,...,n-1$ has a link to buyer $i$. Let $x$ be a large constant.  $V_{1n}=1, V_{11}=\frac{\frac{n}{n-1}x+n}{x+n}-\epsilon, V_{in}=x, V_{ii}=\frac{n^2x}{(x+n)(n-1)}+(i-1)\epsilon, V_{i,i-1}=\frac{n}{n-1}x ,\forall i\in [2,n]$. $\epsilon$ is a very small quantity to break ties. All other valuations are zero.
    
The optimal social welfare when all sellers join platform is $W^{\star}=nx+1$. We will show platform's revenue optimal transaction fee is $\alpha^{\star}=1$ and only seller $n$ joins the platform. Price of anarchy is then $\frac{nx+1}{x+\sum_{i=1}^{n-1}v_{ii}}\rightarrow n$ as $x\rightarrow \infty$. There are two unique properties to this market:
    
\begin{enumerate}
    \item whenever seller $j$ joins, seller $j+1$ joins because its off-platform edge is taken.
    \begin{eqnarray*}
        V_{j+1,j}=\frac{n}{n-1}x\geq V_{jj}+V_{j+1,j+1}, & \forall j=1,2,\ldots n-1
    \end{eqnarray*}
    This reads if $j+1$ does not join the platform, max weight matching matches seller $j$ to buyer $i=j+1$, and seller $j+1$ won't sell. 
    \item whenever all sellers $j'\geq j+1$ joins, seller $j$ joins. To see this, when only seller $n$ is on platform, seller $n-1$ joins when transaction fee $\alpha_{n-1}$ satisfies 
    \begin{eqnarray*}
        \alpha \leq \alpha_{n-1} = 1- \frac{p_{n-1}^{\off}(P=\{n\})}{p_{n-1}^{\on}(P=\{n\})} = 1-V_{n-1,n-1}/\frac{n}{n-1}x
    \end{eqnarray*}
    Now seeing seller $n-1$ joins, seller $n-2$ reasons about threshold of joining $\alpha_{n-2}$
    \begin{eqnarray*}
        \alpha \leq \alpha_{n-2} = 1- \frac{p_{n-2}^{\off}(P=\{n-1,n\})}{p_{n-1}^{\on}(P=\{n-1,n\})} = 1- v_{n-2,n-2}/\frac{n}{n-1}x
    \end{eqnarray*}
    So when $V_{n-2,n-2}< V_{n-1,n-1}$ seller $n-2$ joins whenever $n-1$ joins. This continues with $V_{2,2}\leq V_{3,3}\leq ...\leq V_{n-1}$ until the first seller. When seller $2,3,...,n$ joins, the first seller joins when
    \begin{eqnarray*}
        \alpha \leq \alpha_1 = 1-\frac{p_1^{\off}(P=S\setminus \{1\})}{p_1^\on(P=S\setminus \{1\})} = 1- v_{11}/(\frac{x}{n-1}+1) 
    \end{eqnarray*}
    The claim is true if $\alpha_1\geq \alpha_2 \geq ...\geq \alpha_{n-1}$. Verifying the buyers' valuations $v_{jj}$ for it $j=1,2,...,n-1$ we see \begin{eqnarray}\label{eqn:general_valuation_n_buyer_n_seller}
    V_{11}/(\frac{x}{n-1}+1) < V_{22}/\frac{n}{n-1}x <\ldots < V_{n-1,n-1}/\frac{n}{n-1}x    
    \end{eqnarray}
    This guarantees whenever $j+1$ joins, $j$ joins.
\end{enumerate}

Because of these two properties, only two pure equilibria exists. The platform can either post $\alpha^{\star}=1$ and charge $x$ from seller $n$, or posting $\alpha_1$ and have all sellers join. In the latter case, $p_j^\on=\frac{x}{n-1}+1, \forall j=1,2,...,n-1$
\begin{eqnarray*}
    Rev(\alpha_1) = \alpha_1[p_n^\on(S)+\sum_{j=1}^{n-1}p_j^\on(S)] = [1- v_{11}/(\frac{x}{n-1}+1)](x+n)
\end{eqnarray*}
We can verify $Rev(\alpha =1)=x> Rev(\alpha_1)$. This proves among all pure equilibrium, only seller $n$ joining is the revenue optimal one. We now prove no mixed equilibrium have larger revenue than $x$. If seller $j=1$ mixes, her revenue from not joining equals to that of joining. Off platform price is $v_{11}$. On platform price is at least as large as $\frac{n}{n-1}x-x=\frac{x}{n-1}$ because $x$ is the max externality seller $1$ imposes on others when transacting to buyer 2. This requires
\begin{eqnarray*}
    v_{11} = \alpha p_1^{\on} \geq \alpha\frac{x}{n-1}  & \text{or } & \alpha \leq (n-1)v_{11}/x
\end{eqnarray*}
Then the platform's revenue is no larger than $\alpha n x = n(n-1)v_{11} << x$ when $x\rightarrow \infty$. The same argument is true when seller 2 to n-1 mixes: referral fee $\alpha$ is too small for platform to make any revenue. 

The above verifies the $n$ PoA for a market with $n=m$ sellers and buyers. For a market $n>m$, we can have $m$ buyers with the above valuations and the $n-m$ buyers valuing no sellers, obtaining a $m=\min\{m,n\}$ lower bound. For a market $m>n$, we can have $n$ sellers with the above valuations and the $m-n$ sellers being valued by no buyers, obtaining a $n=\min\{m,n\}$ lower bound.

\end{proof}

\section{Proof of Theorem~\ref{thm:mixed_poa}}\label{app:mixed_proof}

Section~\ref{thm:mixed_poa} explains the intuition for the poa for mixed equilibrium. Here we give a complete and formal proof.  
\mixedPoA*

\begin{proof}
Consider a mixed Platform Equilibrium $\x = (x_1,\ldots, x_m)$ for a buyer-seller network $(S,B,G)$, where $x_j$ is the probability the seller $j$ joins the platform. We define the following Bayesian game:
    \begin{itemize}
        \item For each seller $j$, with probability $x_j$, $j$ can transact with all buyers (``type 1", $t_j=t_1$), and with probability $1-x_j$, it can only transact with the buyers linked to $j$ in $G$ (``type 2", $t_j=t_2$).
        \item The platform posts a transaction-fee $\alpha$.
        \item Before knowing the realization of its type, seller $j$ chooses a pure strategy $a_j\in\{\on,\off\}$ for joining the platform and being able to transact with all buyers, or staying off platform. Let $\mathbf{a}$ denote the joint strategy and $\mathbf{a}_{-j}$ the joint strategy except $j$'s.
        \item Given the realized graph $G'$, a competitive equilibrium is formed. Market clears according to the maximum competitive prices $p$.
        \item A seller has off platform utility $u'_j(\off,\mathbf{a}_{-j};\mathbf{t})=E_{G'\sim \mathbf{a},\mathbf{t}}[p_j]$ and on platform utility $u'_j(\on,\mathbf{a}_{-j};\mathbf{t})=(1-\alpha)E_{G'\sim \mathbf{a},\mathbf{t}}[p_j]$. A seller joining the platform pays an $\alpha$-fraction of their revenue to the platform.
    \end{itemize}

    We now show no seller joining $\mathbf{a}^{\off}=\{\off,\dots,\off\}$ is a pure-strategy Bayes-Nash equilibrium: $$\forall j, E_{\mathbf{t}\sim \x}[u'_j(\off,\mathbf{a}^{\off}_{-j};\mathbf{t})]\geq E_{\mathbf{t}\sim \x}[u'_j(\on,\mathbf{a}^{\off}_{-j};\mathbf{t})]$$ For this, consider a seller j that adopts probability $x_j$ in the mixed Nash Platform Equilibrium of the original, complete information platform game. There are three cases to consider
   
    \noindent\textbf{Case 1:} $x_j=0$ In this case, $j$'s expected utility from not joining the platform \textit{in the original, complete information game} is at least as much as $j$'s expected utility from joining given $\x_{-j}$. \textit{In the Bayesian game}, given that no other seller joins the platform, $j$'s utility from either joining or not joining the platform is exactly $j$'s utility for joining or not joining in the complete information game,
as other links are formed according to $\x_{-j}$.

    \noindent\textbf{Case 2:} $x_j=1$: In this case, in the Bayesian game, $j$ is linked to all buyers with probability $1$. Thus,  joining the platform in the Bayesian game does not increase $j$'s utility.

    \noindent\textbf{Case 3:} $x_j\in (0,1)$: In this case, in the \textit{complete information game}, $j$'s expected utility is the same for joining and not joining given $\x_{-j}$; otherwise, $j$ would deviate and $\x$ would not be an equilibrium. Let $u_j^\x$ denote this quantity. \textit{In the Bayesian game}, first consider $j$'s expected utility where it does not join
    \begin{eqnarray*}
        E_{\mathbf{t}\sim \x}[u'_j(\off,a^{\off}_{-j};\mathbf{t})] &=& (1-x_j)E_{\mathbf{t}_{-j}\sim\x_{-j}}[u'_j(\off,a^{\off}_{-j};t_2,\mathbf{t}_{-j})] + x_j E_{\mathbf{t}_{-j}\sim\x_{-j}}[u'_j(\off,a^{\off}_{-j};t_1,\mathbf{t}_{-j})]\\
        & \geq & (1-x_j)u^{\x}_j + x_j u^{\x}_j = u^{\x}_j
    \end{eqnarray*}
    The first term on the right reads with probability $1-x_j$, $t_j=t_2$ and $j$'s links stay as they were in $G$. Then $j$'s expected utility is exactly $u_j^\x$, as this is $j$'s utility from not joining the platform in the original game when all other sellers' links are formed according to $\x_{-j}$. 
    The second term on the right reads with probability $x_j$, $t_j=t_1$ and all $j$'s links are formed. Then $j$'s expected utility is at least as much as when it joins in the original game given all other sellers' links are formed according to $\x_{-j}$. This is because link formation distribution is the same, while the seller does not have to pay the $\alpha$ fee to the platform in the Bayesian game. Therefore, $j$'s expected utility from not joining is at least $u_j^\x$. 

    Then consider $j$'s expected utility if it joins the platform in the Bayesian game.
$$E_{\mathbf{t}\sim \x}[u'_j(\on,a^{\off}_{-j};\mathbf{t})] = u^{\x}_j$$ It exactly equals to the utility of $j$ joining the platform in the original setup, as the link formation distribution is the same as if $j$ joins in the original setup, and the fee $j$ pays is the same. So again $j$'s expected utility for not joining in the Bayesian game is at least as that for joining.

    In all above cases, each sellers' expected utility of not joining is at least as much as joining when no others join. Thus, $\mathbf{a}^{\off}$ is indeed a pure strategy Bayes-Nash equilibrium. The expected welfare of this equilibrium is exactly the same as the mixed Platform Equilibrium $\x$ in the original complete information game, as we have the same distribution over the formation of links, and the competitive equilibrium formed always maximizes the welfare given the links.
    
    For the last step of the proof, show if no sellers join is an equilibrium in the Bayesian game, then PoA with respect to pure strategies is at most $\frac{2-\alpha}{1-\alpha}$. The proof follows the same arguments of the proof of this case in Theorem~\ref{thm:pure_poa}, but the quantities we reason about for seller $j$ are in expectation over type probability $\x_{-j}$, which is no longer fixed to the mixed equilibrium. We prove $$\frac{W^\star}{E_{G'\sim \mathbf{a}^{\off},\mathbf{t}}[W(S,B,G')]} \leq \frac{2-\alpha}{1-\alpha}$$ where $G'$ is the realized network graph where no seller joins the market in the Bayesiain game. 

    Consider a single seller $j$ joining the platform, denote $G(j)$ as the realized graph where only $j$ joins. The on and off platform utility of seller $j$ can be expressed as 
    \begin{eqnarray*}
        E_{\mathbf{t}\sim\x}[u'_j(\on,\mathbf{a}^{\off}_{-j};\mathbf{t})] & = & (1-\alpha)E_{\mathbf{t}\sim\x}E_{G(j)\sim \mathbf{a}^{\off}_{-j},\mathbf{t}} [p_j] \\
        & = & (1-\alpha) E_{G(j)\sim a^{\off}_{-j},\x}[W(S,B,G(j))-W(S\setminus\{j\},B,G(j))]\\
        E_{\mathbf{t}\sim\x}[u'_j(\off,\mathbf{a}^{\off}_{-j};\mathbf{t})] & = & E_{\mathbf{t}\sim\x}E_{G'\sim a^{\off},\mathbf{t}} [\hat{p}_j] \\
        & = & E_{G'\sim \mathbf{a}^{\off},\x}[W(S,B,G')-W(S\setminus\{j\},B,G')]
    \end{eqnarray*}
    As $j$ does not join the platform, the off platform utility is larger than on utility. Rearranging
    \begin{eqnarray}
        E_{G(j)\sim \mathbf{a}^{\off}_{-j},\x}[W(S,B,G(j))] & \leq & E_{G(j)\sim \mathbf{a}^{\off}_{-j},\x}[W(S\setminus\{j\},B,G(j))]+\frac{1}{1-\alpha}E_{G'\sim \mathbf{a}^{\off},\x}[W(S,B,G')]\nonumber\\
        & - & \frac{1}{1-\alpha}E_{G'\sim \mathbf{a}^{\off},\x}[W(S\setminus\{j\},B,G')]\nonumber\\
        & = & \frac{1}{1-\alpha}E_{G'\sim \mathbf{a}^{\off},\x}[W(S,B,G')] -  \frac{\alpha}{1-\alpha}E_{G'\sim \mathbf{a}^{\off},\x}[W(S\setminus\{j\},B,G')] \label{eq:bound1_mixed}
    \end{eqnarray}

    Let $i^\star(j)$ be the buyer matched to $j$ in the ideal matching $W^\star$. Since $j$ can be matched to $i^\star(j)$ in $G(j)$,
    \begin{eqnarray}
        W(S,B, G(j)) & \geq& v_{i^\star(j)j} + W(S\setminus\{ j\}, B\setminus\{i^\star(j)\}, G(j)) \label{eq:bound2_mixed}
    \end{eqnarray}
    But since $G(j)$ and $G'$ only differ by seller $j$ always joining the platform 
    \begin{eqnarray}
        E_{G(j)\sim \mathbf{a}^{\off}_{-j},\x}[W(S\setminus\{j\},B\setminus\{i^\star(j)\},G(j))]=E_{G'\sim \mathbf{a}^{\off},\x}[W(S\setminus\{j\},B\setminus\{i^\star(j)\},G')] \label{eq:taking_out_j_same_mixed}
    \end{eqnarray}
    
    Combining Equations~\eqref{eq:bound1_mixed}, ~\eqref{eq:bound2_mixed} and ~\eqref{eq:taking_out_j_same_mixed}
    \begin{eqnarray}
        v_{i^\star(j)j} & \leq & \frac{1}{1-\alpha}E_{G'\sim \mathbf{a}^{\off},\x}[W(S,B,G')] -  \frac{\alpha}{1-\alpha}E_{G'\sim \mathbf{a}^{\off},\x}[W(S\setminus\{j\},B,G')] \nonumber\\
        & - & E_{G'\sim \mathbf{a}^{\off},\x}[W(S\setminus\{j\},B\setminus\{i^\star(j)\},G')]\nonumber \\
        &= & E_{G'\sim \mathbf{a}^{\off},\x}[\frac{1}{1-\alpha}W(S,B,G')-\frac{\alpha}{1-\alpha}W(S\setminus\{j\},B,G')-W(S\setminus\{j\},B\setminus\{i^{\star}(j)\},G')]
        \label{eq:bound3_mixed}
    \end{eqnarray}

    For a fixed realized graph $G'$, the three terms inside the expectation in Eq.~\eqref{eq:bound3_mixed} have been analyzed in Eq.~\eqref{eq:bound3}. Directly taking the results there we have a similar form to Eq. ~\eqref{eq:combine_three_terms_right}
    \begin{eqnarray*}
        v_{{i^\star}(j)j} &\leq &  E_{G'\sim \mathbf{a}^{\off},\x}[\frac{1}{1-\alpha}\cdot v_{i^{G'}(j)j} + v_{i^\star(j)t(j)}]
    \end{eqnarray*}
    Again $t(j)$ is the seller who transacts with $i^{\star}(j)$ in $G'$, also named as \textit{ the twin of $j$ }. 

    Summing over all sellers $j$, we get
    \begin{eqnarray*}
        W^\star =  \sum_j v_{{i^\star}(j)j} & \leq & \sum_j E_{G'\sim \mathbf{a}^{\off},\x}\left(\frac{1}{1-\alpha}\cdot v_{i^{G'}(j)j} + v_{i^\star(j)t(j)}\right)\\ 
        & = &  E_{G'\sim \mathbf{a}^{\off},\x} [\frac{1}{1-\alpha}\cdot\sum_j v_{i^{G'}(j)j} + \sum_j v_{i^\star(j)t(j)}]\\
        & \leq & E_{G'\sim \mathbf{a}^{\off},\x} [\frac{1}{1-\alpha} W(S,B,G') + \sum_j v_{i^{G'}(j)j}]\\
        & = & \frac{2-\alpha}{1-\alpha} E_{G'\sim \mathbf{a}^{\off},\x}[W(S,B,G')],
    \end{eqnarray*}
    where the last inequality follows because each seller $j$ has a distinct twin $t(j)$, so we sum over distinct edges of $G'$.
\end{proof}

\section{Generalizations}\label{app:generalizations}
\subsection{Market with Multiple Platforms}\label{app:multiple_platforms}
\MultiPlatform*
\begin{proof}[proof sketch]
    The proof follows from the same argument as that of Theorem~\ref{thm:pure_poa} and Theorem~\ref{thm:mixed_poa}. For the Price of Anarchy of pure Platform Equilibrium, we again first prove if no seller chooses to join any platform, welfare is guaranteed by, $$ \frac{W^\star}{W(S,B,G)}\leq\frac{2-\alpha}{1-\alpha}.$$
   
    Let $i^{\star}(j)$ be the buyer that $j$ transacts with in ideal matching $W^\star$, $f^\star(j)$ be the platform that $j$ pays in ideal matching, and $G^{\star}(j)$ be the network when only seller $j$ joins platform $f^\star(j)$. Then inequalities~\ref{eq:p_on} to \ref{eq:bound2} still hold after changing $W(S,B,G(j))$ for $(W,S,B,G^\star(j))$, and inequalities~\ref{eq:bound3} to \ref{eq:combine_three_terms_right} remain the same. Now consider any pure equilibrium where a set of sellers $P$ join some platforms. Denote the resulting graph by $G(P)$.
    Again we can construct another market where the buyers have the same valuations as the original market but the initial network is $G'=G(P)$. No seller joining is a pure Platform Equilibrium because otherwise some sellers would have deviated in the pure equilibrium in the original market and joined some other platforms. So $\frac{W^\star}{W(S,B,G(P))}=\frac{W^\star}{W(S,B,G')}\leq \frac{2-\alpha}{1-\alpha}$. The proof for the Price of Anarchy of mixed Platform Equilibrium follows similarly.
\end{proof}

\subsection{Sellers with Production Costs} \label{app:seller_with_cost}
In this section, we extend our results to sellers $j\in S=\{1,2,...,m\}$ with production cost $c_j\geq 0$. Sellers with a cost too high barely trade, and do not contribute to social welfare. Therefore, for the discussion we assume all sellers transact. This can be modeled as each seller $j$ being connected off-platform to a dummy buyer $i_j$ that only values seller $j$ and dislikes other sellers: 
$$v_{i_j,j'}= 
\begin{cases}
    c_j,     & \text{if } j'=j\\
    -\infty,        & \text{otherwise}
\end{cases}.$$

For any market $M=(S,B,\mathbf{v},\mathbf{c},G)$ with cost $\mathbf{c}$, define a many-to-one mapping of markets $T(M)=M'$, where $M'$ is a new market with zero seller cost and buyer valuations 
$$v'_{ij}= 
\begin{cases}
    v_{ij}-c_j,     & \text{if } v_{ij}\geq c_j\\
    -\infty,        & \text{otherwise}
\end{cases}$$

\begin{lemma}
\label{lem:ce_with_cost_exist}
    If $(\mathbf{p},\mathbf{a})$ is a competitive equilibrium in market $M$, define $p_j'=\max\{p_j-c_j,0\}$, then $(\mathbf{p}',\mathbf{a})$ is a competitive equilibrium in market $M'$. If $(\mathbf{p}',\mathbf{a})$ is a competitive equilibrium in market $M'$, define $p_j=p'_j+c_j$, then $(\mathbf{p},\mathbf{a})$ is a competitive equilibrium in market $M$.
\end{lemma}

\begin{proof}
We first prove if $(\mathbf{p}',\mathbf{a})$ is a competitive equilibrium in $M'$, then $(\mathbf{p},\mathbf{a})$ is a competitive equilibrium in $M$. 
The first three conditions for competitive equilibrium in Definition~\ref{def:comp_eq} are easy to check. Without loss of generality each seller $j$ in $M'$ transacts. Because otherwise $p_j=0$ equals to the dummy buyers' valuation in $M'$: $v'_{i_j,j}=0$. So the last condition is satisfied. We now show every buyer gets their most preferred outcome with $(\mathbf{p},\mathbf{a})$. When buyer $i$ does not obtain an item, $\forall j, p'_j\geq v'_{ij}$. If $v_{ij}\geq c_j$ expanding we have $p_j\geq v_{ij}$; If $v_{ij}< c_j$ with prices being non-negative $p_j=p'_j+c_j\geq c_j > v_{ij}$. In either case buyer $i$'s best choice in $M$ is not obtaining an item. When buyer $i$ does obtain an item $j^{\star}$ in $M'$, $v_{ij^\star}\geq c_{j^\star}$, and $\forall j, v'_{ij^\star}-p'_{j^\star}=v_{ij^\star}-p_{j^\star}\geq v'_{ij}-p'_j$. If $v_{ij}\geq c_j$ then $v_{ij^\star}-p_j^\star\geq v_{ij}-p_j$ and we are done; otherwise $v_{ij}-p_j<c_j-p_j \leq 0 \leq v_{ij^\star}-p_j^\star$. So indeed buyer $i$'s favorite item in $M$ is $j^\star$.

Now prove the other direction: if $(\mathbf{p},\mathbf{a})$ is a competitive equilibrium in $M$, $(\mathbf{p}',\mathbf{a})$ is a competitive equilibrium in $M'$. Again the first three criteria for Definition~\ref{def:comp_eq} are easy to verify. With dummy buyers in $M$, all items are sold in $\mathbf{a}$ so the last criteria is satisfied.
We now verify each buyer gets their most preferred outcome in $M'$. When buyer $i$ does not obtain an item, $\forall j, v_{ij}\leq p_j$. If $v_{ij}< c_j$ then $v'_{ij}=-\infty <0 \leq p'_j$; if $v_{ij}\geq c_j$ then expanding $v'_{ij}\leq p'_j$. In either case buyer $i$'s best choice in $M'$ is not obtaining an item.
If a dummy buyer $i_j$ acquires an item $j^\star=j$ in $M$. Then $v_{i_j,j^{\star}}=c_{j^\star}\geq p_{j^\star}$. With the same allocation $v'_{ij^\star}-p'_{j^\star}=0-\max\{0,p_{j^\star}-c_{j^\star}\}=0$. But for any other seller item $j'$, $v'_{i_j,j'}=-\infty$ so the dummy buyer strictly does not want other items. If buyer $i$ that is not a dummy buyer obtains an item $j^\star$ in $M$, then $p_{j^\star}\geq c_{j^\star}$ because otherwise dummy buyer $i_{j^\star}$ does not satisfy the demand. 
By non-negative of utility and price $v_{ij^\star}\geq p_{j^\star} \geq c_{j^\star}$ so $\forall j, v'_{ij^\star}-p'_{j^\star}=v_{ij^\star}-p_{j^\star}\geq v_{ij}-p_{j}$. 
If $v_{ij}\geq c_j$ and $p_j\geq c_j$ then we have $v'_{ij} - p'_j =v_{ij}-c_j-\max\{0,p_j-c_j\}=v_{ij}-p_j$ and we are done, if $v_{ij}\geq c_j$ and $p_j< c_j$ then we have  $v'_{ij} - p'_j =v_{ij}-c_j< v_{ij}-p_j$ and we are again done; 
if $v_{ij}< c_j$ then $v'_{ij}=-\infty$ and surely in $M'$ buyer $i$ does not demand $j$.%: $v'_{ij^\star}-p'_{j^\star}=v_{ij^\star}-p_{j^\star}\geq 0 > v'_{ij}-p'_j$.
\end{proof}

The above lemma gives a one-to-one correspondence between competitive equilibrium in $M$ and $M'$. This allows us to analyze social welfare in $M$. The social welfare of an allocation $\mathbf{a}$ in $M$ is equal to valuation minus cost in $M'$, that is, $$\sum_{ij}a_{ij}(v_{ij}-c_{j}).$$  By applying the first welfare theorem for the market without costs, we get the following. 
\begin{corollary}
\label{cor:first_welfare_cost}
    In a competitive equilibrium for a market with production costs, the social welfare is maximized with respect to the set of allocations that respect the transaction constraints posed by $G$.
\end{corollary}
\begin{proof}
    By First Welfare Theorem~\ref{thm:first_welfare} in market without cost, the competitive allocation $\mathbf{a}$ in market $M'$ satisfies $\forall \mathbf{a}', \sum_{ij}a_{ij}v'_{ij}\geq \sum_{ij}a'_{ij}v'_{ij}$. If $a_{ij}=1$, by non-negative utility $v_{ij}\geq c_j$. So $\forall \mathbf{a'}, \sum_{ij}a_{ij}(v_{ij}-c_j)\geq \sum_{ij}a'_{ij}v'_{ij}$. For allocations $\mathbf{a'}$ that let buyers purchase despite $v_{ij}<c_j$, the social welfare is smaller than simply letting this transaction opportunity go. So we focus on $\mathbf{a'}$ that only transacts when $v_{ij}\geq c_j$. For these $\mathbf{a'}$,  $\sum_{ij}a_{ij}(v_{ij}-c_j)\geq \sum_{ij}a'_{ij}(v_{ij}-c_j)$. As social welfare incorporates cost in $M$, $\mathbf{a}$ maximizes social welfare.
\end{proof}

Let $W(S,B,\mathbf{v},\mathbf{c},G)$ denote optimal welfare in $M$. For any competitive equilibrium allocation $\mathbf{a}$, it maximizes social welfare in both $M$ and $M'$. By the construction of buyer valuation in $M'$, we observe the following.

\begin{corollary}
\label{cor:equal_welfare}
    Optimal welfare in $M$ and $M'$ are equal. $W(S,B,\mathbf{v},\mathbf{c},G)=W(S,B,\mathbf{v'},G)$.\label{thm:equal_welfare}
\end{corollary}
\begin{proof}
    Consider a competitive equilibrium allocation $\mathbf{a}$ for both market $M$ and $M'$. Optimal social welfare in $M'$ equals to $W(S,B,\mathbf{v}',G)=\sum_{ij}a_{ij}v'_{ij}=\sum_{ij}a_{ij}(v_{ij}-c_j)=W(S,B,\mathbf{v},\mathbf{c},G)$. The second inequality is because the First Welfare Theorem~\ref{thm:first_welfare} requires $\mathbf{a}$ to maximize welfare: whenever $a_{ij}=1, v'_{ij}\geq 0$.
\end{proof}

The two corollaries above are implications of Lemma~\ref{lem:ce_with_cost_exist} on the allocation side of competitive equilibrium. On the price side, though a competitive price $p_j$ in $M$ associates with $p'_j=\max\{p_j-c,0\}$, the maximum competitive price correspondence is simpler.

\begin{corollary}
\label{cor:max_price_cost}
    Given a of price $\mathbf{p}$, define $p'_j=p_j-c_j$. $\mathbf{p}$ is the max competitive price for $M$ if and only if $\mathbf{p}'$ is the max competitive price for $M'$.
\end{corollary}
\begin{proof}
    If $\mathbf{p}$ is the maximum competitive price in $M$, $p_j\geq c_j$: if $j$ is matched to a dummy buyer $i_j$ in equilibrium, the price can always increase until $p_j=c_j$; if $j$ is matched to other buyers, then $p_j>c_j$ because otherwise $i_j$ envies. By Lemma~\ref{lem:ce_with_cost_exist}, $p'_j=p_j-c_j$ is a competitive price in $M'$. If in $M'$ there is a higher competitive price $\mathbf{p}''$, by Lemma~\ref{lem:ce_with_cost_exist} in market $M$ there must be a competitive price $\mathbf{p}''+\mathbf{c}$ higher than $\mathbf{p}$, contradict. So $\mathbf{p}'$ is indeed the maximum competitive price. The other direction works similarly.
\end{proof}

We are now ready to inspect the effect of cost on sellers' decisions to join platform with maximum competitive price. If platform would charge a transaction fee as a percentage of sellers' utility, $p_j-c_j$, all our previous results without cost immediately hold. As the platform cannot directly observer sellers' costs, the transaction fee is a percentage sellers' price, which distorts sellers' incentives to join. At fee $\alpha$, seller $j$ in $M$ joins the platform if $(1-\alpha)p_j^{\on}-c_j \geq p_j^{\off}-c_j\Leftrightarrow (1-\alpha)(p_j^{'\on}+c_j) \geq p_j^{'\off}+c_j \Leftrightarrow (1-\alpha)p_j^{'\on}\geq p_j^{'\off}+\alpha c_j$. The higher the cost, the less likely seller $j$ joins the platform. We generalize our result in Section~\ref{sec:poa-regulated} to settings with costs.

\purePoACost*
\begin{proof}
    The proof logic is similar to that in section~\ref{sec:poa-regulated}, where we first prove for pure equilibrium then extend it to mixed equilibrium. If no seller joins the market at $\alpha$, then $\forall j\in P$
    $$p_{j}^{'\off} +\alpha c_j \geq (1-\alpha)p_{j}^{'\on}$$ Expanding the price 
    \begin{eqnarray*}
        W(S,B,\mathbf{v}',G(j)) - W(S,B,\mathbf{v}',G) & \leq & p^{'\on}_j-p^{'\off}_j \leq  \alpha[p^{'\on}_j+c_j] \\
        & \leq & \frac{\alpha}{1-\alpha}\cdot p^{'\off}_j + \frac{\alpha}{1-\alpha}c_j\\
        & \leq & \frac{\alpha}{1-\alpha}\left(W(S,B,\mathbf{v}',G) - W(S\setminus\{j\},B,\mathbf{v}',G)\right)+\frac{\alpha}{1-\alpha}c_j
    \end{eqnarray*}

    Rearranging gives 
    \begin{eqnarray*}
    W(S,B,\mathbf{v}',G(j)) \leq
        \frac{1}{1-\alpha}W(S,B,\mathbf{v}',G) - \frac{\alpha}{1-\alpha}W(S\setminus\{j\},B,\mathbf{v}',G)+\frac{\alpha}{1-\alpha}c_j
    \end{eqnarray*}

    Let $i^\star(j)$ be the buyer matched to $j$ in the ideal matching $W^\star(S,B,\mathbf{v},\mathbf{c},G)$. By Lemma~\ref{lem:ce_with_cost_exist}, $i^\star(j)$ is the buyer matched to $j$ in $W^\star(S,B,\mathbf{v}',G)$. Since matching $i^\star(j)$ to $j$ is also one of the options in $(S,B,\mathbf{v}',G(j))$
    \begin{eqnarray*}
        W(S,B,\mathbf{v}',G(j)) & \geq& v'_{i^\star(j)j} + W(S\setminus\{ j\}, B\setminus\{i^\star(j)\}, \mathbf{v}',G)
    \end{eqnarray*}

    Combining the two above inequalities
    \begin{eqnarray}
        v'_{i^\star(j)j} \leq \frac{1}{1-\alpha}W(S,B,\mathbf{v}',G)-\frac{\alpha}{1-\alpha}W(S\setminus\{j\},B,\mathbf{v}',G)-W(S\setminus\{j\},B\setminus\{i^\star(j)\},\mathbf{v}',G)+\frac{\alpha}{1-\alpha}c_j \label{eq:bound3_cost}
    \end{eqnarray}

    We wish to relate the right-hand side of Eq.~\eqref{eq:bound3_cost} to terms that relate to $W(S,B, \mathbf{v}',G)$ and $W^\star$.
 Let $i^G(j)$ be the buyer matched to $j$ in $W(S,B, \mathbf{v}',G)$ and $W(S,B,\mathbf{v},\mathbf{c},G)$. First, consider $W(S\setminus\{j\},B,\mathbf{v}',G)$. By 
definition,
    \begin{eqnarray*}
        W(S,B,
        \mathbf{v}',G) = W(S\setminus\{j\}, B\setminus\{i^G(j)\},\mathbf{v}', G) + v'_{i^G(j)j}. 
    \end{eqnarray*}
    Thus, we have,
    \begin{eqnarray}
        W(S\setminus\{j\},B, \mathbf{v}',G) \geq W(S\setminus\{j\}, B\setminus\{i^G(j)\}, \mathbf{v}', G) = W(S,B,\mathbf{v}',G)-v'_{i^G(j)j}.\label{eq:bound4_cost}
    \end{eqnarray}

    As for $W(S\setminus\{j\},B\setminus\{i^\star(j)\}, \mathbf{v}',G)$, seller  $j$ is matched to $i^G(j)$ in $G$, while $i^\star(j)$ is matched to some potentially different vertex in $G$, which we call \textit{the twin of $j$}
 and denote by $t(j)$. We  have the following inequality,
    \begin{eqnarray*}
        W(S,B,\mathbf{v}',G) &\leq & W(S\setminus\{j,t(j)\},B\setminus\{i^G(j),i^\star(j)\},\mathbf{v}',G) + v'_{i^G(j)j} + v'_{i^\star(j)t(j)}\nonumber \\
        &\leq & W(S\setminus\{j\},B\setminus\{i^\star(j)\},\mathbf{v}',G) + v'_{i^G(j)j} + v'_{i^\star(j)t(j)}, 
    \end{eqnarray*}
    where the first inequality is an equality if $j\neq t(j)$. Rearranging gives
    \begin{eqnarray}
        W(S\setminus\{j\},B\setminus\{i^\star(j)\},\mathbf{v}',G) \geq W(S,B,\mathbf{v}',G) - (v'_{i^G(j)j} + v'_{i^\star(j)t(j)}).\label{eq:bound5_cost}
    \end{eqnarray}

    Combining Equations~\eqref{eq:bound3_cost},~\eqref{eq:bound4_cost}, and~\eqref{eq:bound5_cost}, we get
    \begin{eqnarray}
        v'_{{i^\star}(j)j} &\leq& \frac{1}{1-\alpha}W(S,B,\mathbf{v}',G)-\frac{\alpha}{1-\alpha}\left(W(S,B,\mathbf{v}',G)-v'_{i^G(j)j}\right)\nonumber\\  &-& \left(W(S,B,\mathbf{v}',G) - (v'_{i^G(j)j} + v'_{i^\star(j)t(j)})\right)+\frac{\alpha}{1-\alpha}c_j \nonumber \\
        & = & \frac{1}{1-\alpha}\cdot v'_{i^G(j)j} + v'_{i^\star(j)t(j)} + \frac{\alpha}{1-\alpha}c_j \nonumber%\label{eq:combine_three_terms_right_cost}
    \end{eqnarray}

    Summing over all sellers $j$, we have
    \begin{eqnarray}
        W^\star(S,B,\mathbf{v},\mathbf{c},G)=W^\star(S,B,\mathbf{v}',G)   
        = \sum_j v'_{{i^\star}(j)j} & \leq & \sum_j \left(\frac{1}{1-\alpha}\cdot v'_{i^G(j)j} + v'_{i^\star(j)t(j)}+\frac{\alpha}{1-\alpha}c_j \right)\nonumber\\ & = &  \frac{1}{1-\alpha}\cdot\sum_j v'_{i^G(j)j} + \sum_j v'_{i^\star(j)t(j)}+\frac{\alpha}{1-\alpha}\sum_j c_j\nonumber\\
        & \leq & \frac{2-\alpha}{1-\alpha} W(S,B,\mathbf{v}',G)+\frac{\alpha}{1-\alpha}\sum_j c_j\nonumber\\
        & = & \frac{2-\alpha}{1-\alpha} W(S,B,\mathbf{v},\mathbf{c},G)+\frac{\alpha}{1-\alpha}\sum_j c_j \label{eq:poa_cost}
    \end{eqnarray}
    where the first and the last equality follows from Corollary~\ref{cor:equal_welfare}, and the last inequality is because each seller $j$ has a distinct twin $t(j)$, so we sum over distinct edges of $G$. Bring in $\sum_j c_j =\beta W^\star(S,B,\mathbf{v},\mathbf{c},G)$ 
    \begin{eqnarray}
        W(S,B,\mathbf{v},\mathbf{c},G) \geq \frac{1-\alpha-\alpha\beta}{2-\alpha} W^\star(S,B,\mathbf{v},\mathbf{c},G)
    \end{eqnarray}

    Now if a set $P$ sellers join the platform in a pure equilibrium, just like that in the proof for Theorem~\ref{thm:pure_poa}, we can construct another market where the initial network is $G'=G(P)$. For the same transaction fee, no seller joining the platform in $(S,B,\mathbf{v},\mathbf{c},G')$ is an equilibrium. And we can apply inequality~\ref{eq:poa_cost} to $G'$ and achieve the same result.
    
    The extension to mixed equilibrium works similarly to that in the proof for Theorem~\ref{thm:mixed_poa} and is omitted.
\end{proof}

The above theorem implies that at the same transaction fee, the social welfare guarantee (naturally) deteriorates as costs get higher. This was specially relevant during the COVID-19 crisis, where demand and supply surged on digital platforms \citep{raj2020covid} and the cost of production grew \citep{felix2020us}.
At $30\%$ of transaction fee (Table~\ref{table:platforms and their coommission rate}) and $30\%$ of food cost \footnote{Most profitable restaurants aim for a food cost percentage between 28 and 35\%. This does not include other cost factors such as labor, rentals, etc. Figure is taken in 2023 from Doordash website \url{https://get.doordash.com/en-us/blog/food-cost-percentage}}, the resulting welfare at a  Platform Equilibrium  is at least $23.5\%$ of the ideal welfare. This is smaller than the $41\%$ without costs, further demonstration the need to regulate transaction fee during periods where production costs increase.

\end{document}